\documentclass{article}
\usepackage{epsfig}
\usepackage{float}
\usepackage{amsmath}
\usepackage{graphicx}
\usepackage{latexsym}

\newcommand{\be}{\begin{equation}}
\newcommand{\ee}{\end{equation}}

\newtheorem{theorem}{Theorem}[section]
\newtheorem{lemma}[theorem]{Lemma}
\newtheorem{corollary}[theorem]{Corollary}
\newtheorem{proposition}[theorem]{Proposition}

\newtheorem{remark}[theorem]{Remark}

\newenvironment{proof}{
    \noindent {\it Proof.}}{\hfill$\Box$
}

\begin{document}

\title{ \bf Bistable reaction-diffusion on a network}

\author{J.-G. Caputo$^1$, 
G. Cruz-Pacheco$^2$ and 
P. Panayotaros $^2$ \\
{ \small $^1$
Laboratoire de Mathematiques,
INSA de Rouen,} \\
{\small Av. de l'Universite,
76801 Saint-Etienne du Rouvray, France} \\
{\small $^2$
Depto. Matem\'{a}ticas y Mec\'{a}nica,
I.I.M.A.S.-U.N.A.M., } \\
{\small Apdo. Postal 20--726, 01000 M\'{e}xico D.F., M\'{e}xico } }

\maketitle

\date{\ }

\begin{abstract}
We study analytically and numerically a bistable reaction-diffusion
equation on an arbitrary finite network. We prove that stable fixed points 
(multi-fronts) exist for any configuration as long as the diffusion 
is small. We also study fold bifurcations leading to  
depinning and give a simple depinning criterion.
These results are confirmed
by using continuation techniques from bifurcation theory and by
solving the time dependent problem near the treshold.
A qualitative comparison principle 
is proved and verified for time dependent solutions, and for 
some related models. 
\end{abstract}

\section{Introduction}

Discrete reaction-diffusion equations arise in many different fields. For
example they
can describe the propagation of a nerve impulse in a neuron \cite{scott}
or the motion of a dislocation \cite{nabarro}. The solutions of
these equations are typically fronts connecting two regions of constant
value, say 0 and 1.  Front pinning and propagation has been studied by 
many authors for 
a one dimensional network for a bistable cubic reaction term. An important 
result obtained by Keener\cite{keener87} is that when the Laplacian is weak,
any arbitrary configuration of 0's and 1's leads to a stable static solution.
The study was extended by Erneux and Nicolis\cite{EN93} who explicitely 
calculated these fronts and gave a pinning criterion. For material
science applications and in the presence of an external forcing, Carpio 
and Bonilla \cite{CB03} gave pinning conditions and estimated the front
speed. For a two dimensional regular lattice, front propagation
was studied by Hoffman and Mallet-Paret\cite{HM10}.

The present article considers arbitrary but finite networks, where to our 
knowledge there are no works. 
We address specifically this problem and study analytically
and numerically static fronts and how they destabilize in  
an arbitrary finite network (graph). The reaction term we use 
is the bistable cubic nonlinearity and the diffusion term
is the standard graph Laplacian of the network (see e.g. \cite{crs01}). 
Throughout the article, we refer to this equation as the Zeldovich model.
We introduce and motivate the bistable reaction-diffusion 
system by considering how an epidemic propagates on a network. 
To describe how the epidemic front moves on the network, we extend
the standard Kermack-McKendrick model (see e.g. \cite{CDEMLPAV09} for a 
recent application) to a network and show how it reduces to a 
discrete Fisher equation. In contrast to the ODE model, the network 
Kermack-McKendrick model
is not commonly used to describe the spread of an epidemic.
The Fisher model only describes the propagation phase.  
The related Zeldovich model we propose is also new
but its cubic bistable nonlinearity 
has a local excitation threshold, which may be a desirable feature  
for both geographic networks, where the epidemic spreads 
from one location to another, and 
agent-based networks, where the disease spreads from one individual
to another.

A first result is the existence of static stable fronts 
for small diffusivity. 
The argument combines the implicit 
function theorem (as in the anticontinuous limit used for
other lattice problems, see \cite{MA94}) with small
diffusivity asymptotics for the front amplitudes.
The proof also uses a suitable definition for the interface
between the active and quiescent sites. 
The statement is analogous to Keener's result for the integer
lattice \cite{keener87}.
We also show that for large diffusivity the only static 
solutions are spatially homogeneous.

The existence of these fronts depends on the diffusivity, 
the nonlinearity, and the local excitation threshold parameters of the model.
We focus on the dependence of the static fronts on the 
diffusivity using numerical continuation techniques. 
The continuation exhibits the fold structure seen in 
one dimensional studies \cite{EN93}.
For general networks the depinning diffusivity threshold depends on the
front configuration, and a static configuration that becomes unstable 
can be pinned elsewhere. We compute numerically the
depinning thresholds for different static solutions and show that
they can be predicted accurately by a simple heuristic expression derived 
for small diffusivity.
By solving the time dependant problem, we verify these findings
and see how the connectivity of the network
affects the propagation of the fronts above the threshold.

We also obtain qualitative comparison results between different solutions
of the Zeldovich equation, showing in particular that "large" fronts
involving large regions of 1's dominate "small" fronts.
Our study also contains comparison results showing that the Fisher 
equation describes faster front propagations than 
both the Zeldovich and Kermack-McKendrick equations. These results 
are also verified numerically. We see also that the Fisher 
and Kermack-McKendrick fronts propagate at comparable speeds and 
are much faster that the Zeldovich fronts. 
Finally we present numerical results for larger local escitation 
threshold
parameters, showing that the static fronts become wider and travel
much faster accross the network when they destabilize. \\
The article is organized as follows. In section 2 we introduce
the Zeldovich equation and discuss the other models. 
Section 3 studies the fixed points of the
Zeldovich equation, presenting theretical and numerical continuation results, 
as well as a depinning criterion.
Section 4 describes comparison results between the solutions of
the Zeldovich equation, and between solutions of the 
Zeldovich, Fisher, and Kermack-McKendrick equations 
Section 5 presents numerical
results of the evolution problem; there we validate the pinning
thershold for different fronts and compare the dynamics of
large and small fronts. We also show that fronts become wider
as the nonlinearity treshold increases and we compute the pinning
treshold. Conclusions are given in section 6.

%of different configurations under the Zeldovich system. We see 
%that the partial order between 
%configuration defined by comparing values at all sites
%of the network is preserved by the Zeldovich evolution. Similar 
%comparison theorems imply that the Fisher system evolves faster 
%than both the Zeldovich and Kermack-McKendrick systems with the same 
%diffusivity parameters. We see numerically that this result is more relevant
%to the Fisher, and Kermack-McKendrick systems that describe the spreading
%of the epidemic to the whole network at comparable
%timescales. The Zeldovich system describes a much slower front 
%propagation, and its use as substitute for the other two models,  
%would at least require the use of different parameters. 

\section{The Zeldovich model and epidemic propagation}

One of the main models to describe the time evolution of the
outbreak of an epidemic is the Kermack-McKendrick model\cite{km27}
\begin{eqnarray} \label{km}
S_t = -\alpha S I,\\
I_t = \alpha S I -\beta I,\\
R_t = \beta I,
\end{eqnarray} 
where $S,I,R$ are respectively the number of people susceptible to be infected,
the number of infected and the number of recovered in a total 
constant population $N$.
We have of course
$$S+I+R=N  . $$
The dynamics of the model is that 
$ I_t>0 $  (resp. $ I_t<0 $) if $S > \beta / \alpha$ 
(resp. $S < \beta / \alpha$) .
We also can compute the "final" state of $S$ after the outbreak
$$ S(t) = s(0) \exp \left ( -\alpha \int_0^t I(t') dt' \right ) ~~. $$
 
Roughly speaking, assuming that $I(0)$ is near zero, and 
$S(0) > \beta / \alpha $, the infected population 
$I(t)$ increases, reaches a maximum value and decreases to zero. The 
main questions are 
that maximum value of $I$, the time to reach its, the integral of $I$, etc.

We rescale the variables by $N$ 
$$s = S/N,~i=I/N,~r=R/N  ~ .$$
This yields the system
\begin{eqnarray} \label{km2}
s_t = -\alpha N s i,\\
i_t = \alpha N s i -\beta i,\\
r_t = \beta i ~~.
\end{eqnarray}
We introduce now the possibility of dispersion from city to city with a 
Laplacian term. The system (\ref{km2}) becomes 
% {\bf km3}
\begin{eqnarray} \label{km3}
s_t = \epsilon \Delta s -\gamma s i, \notag \\
i_t = \epsilon \Delta i  +\gamma s i -\beta i, 
\end{eqnarray}
where $\gamma = \alpha N$ and 
the third equation is omitted because of the conservation 
\be\label{sri} s+r+i=1 .  \ee
This model will describe the outbreak of the epidemic, its spreading, 
and eventual
demise as $i$ peaks and starts decreasing at each site.

To simplify even more the model
and get analytical results we only consider the maximum outbreak by
eliminating the $\beta$ term  and only considering the equation
for $i$. If $\beta=0$ then $s+i$ verifies
$$ (s+i)_t = \epsilon \Delta (s+i)$$
so that $s+i$ goes to a constant which we can assume to be 1. 
Then from (\ref{km3}) for $x = i$, we get the Fisher equation 
%{\bf fisher}
\be\label{fisher}
x_t = \epsilon \Delta x  + \gamma (1-x) x, ~\ee
This equation has two homogeneous solutions $x^*=0,1$ and the former
is unstable. The model does not have a treshold as opposed to
the Kermack-McKendrick.
To re-introduce this important feature, we modify the nonlinearity 
into the cubic (Zeldovich) so that we get
%{\bf zeldovich}
\be\label{zeldovich}
x_t = \epsilon \Delta x  + \gamma  (1-x) x (x-a).~\ee
For this, there are only two stable homogeneous solutions
$x^* =0,1$. As discussed in the introduction, this equation has
many physical applications; it is then an important physical model.

If we had a spatially uniform domain 
the term $\Delta$ would be the usual Laplacian. Here we consider an
arbitrary graph, for example the network of six major cities in
Mexico shown in Fig. \ref{g6}. 
Here the nodes correspond to the cities and 
the links correspond to the main roads connecting these cities. 
\begin{figure} [H]
\centerline{\epsfig{file=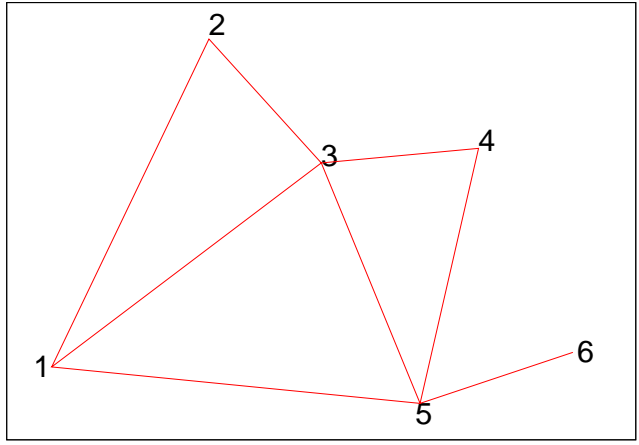,width=0.8\linewidth,angle=0}}
\caption{ Graph of the six main cities in Mexico numbered from 1 to 6:
Guadalajara, Zacatecas, Queretaro, Pachuca,  Mexico City, Puebla.  
The links represent the main roads connecting these cities.
}
\label{g6}
\end{figure}
For this particular example, the term $\Delta$ is 
\be\label{lap_6}
\epsilon \Delta \equiv \epsilon \left(
\begin{array}{cccccc}
-3& 1  & 1 &  0&  1& 0  \\
1 & -2 & 1 &  0&  0& 0 \\
1 & 1  & -4&  1&  1& 0 \\
0 & 0  & 1 & -2&  1& 0 \\
1 & 0  & 1 &  1& -4& 1 \\
0 & 0  & 0 &  0&  1& -1
\end{array}
\right).
\ee
Note that the graph Laplacian $\Delta$ is a non-negative symmetric matrix 
\cite{crs01}.  We use this property below. 
In physical units the parameter $\epsilon$ is
\be\label{epsilon}
\epsilon = \frac{D}{h^2 },
\ee
where $D$ is a diffusion coefficient and $h$ is a typical distance
between cities. The typical time for the diffusion is then
\be\label{teps}
t = \frac{1}{\epsilon}= \frac{h^2}{D}.\ee
At this time we assumed the same diffusion coefficient (weight) for all the
links of the network. If a node is more or less remote from its 
neighbors than the other nodes, then one could modify the weight 
accordingly. With this generalization, we would still have a positive
symmetric graph Laplacian.

Let $\tau$ be the triangle $\{ (s, i) \in [0,1]^2,~~  s + i \leq 1 \}$. 
We have the following result.
\begin{lemma}
\label{invariance-of-unit-cube}
The unit cube $[0,1]^N$ 
is invariant under the evolution of the Zeldovich 
(\ref{zeldovich}) and Fisher equations in ${\bf R}^N$.
The product of the triangles   $\tau^N $ is 
invariant under the Kermack-McKendrick (\ref{km3}) system
with in ${\bf R}^{2N}$.
\end{lemma}
The lemma follows from Propositions
\ref{compare-Zel-Zel}, \ref{compare-Fis-KmK}
in section 4 below 
(these do not use any of the results of section 3).
It is also easy to show that  
the corresponding vector fields at the 
point inwards at the boundaries.

\section{Fixed points of the Zeldovich model}

We want to describe a situation where only some nodes are excited; in the epidemic context, it means that some nodes are infected
and the rest are susceptible. Only the Zeldovich model
(\ref{zeldovich}) has such stable fixed points; these are generalized 
static ``fronts'' 
where some nodes are close to one and the rest close to zero.
Therefore, in this section, we concentrate on the fixed points of the 
Zeldovich model (\ref{zeldovich}). We will clarify the situation for
the Fisher model (\ref{fisher}) below and show why it is less
interesting. For definiteness, throughout this section, we consider 
the 6 node graph from Fig. \ref{g6}; it is clear that the results
can be extended to an arbitrary finite graph.

The fixed point equation we solve is 
\be 
\label{stat-zel} 
F(x,\epsilon) = 0, 
\quad x = [x_1, \ldots, x_n]^T, \ee 
where 
\be 
F_k(x,\epsilon) =  \epsilon (\Delta x)_k  + f_k(x), \quad\hbox{with} 
\ee
\be
\label{fp-zel}
f_k(x) = \gamma  (1-x_k) x_k (x_k-a), 
\quad k = 1, \ldots, n, 
\ee
$\Delta$ is the graph Laplacian of (\ref{lap_6}), 
and $0 < a < 1 $, $\gamma = 1$. 
We will examine how
the fixed points depend on the coupling parameter $\epsilon \geq 0$.

For $\epsilon =  0$, and every partition of the 
set of nodes into three subsets $S_0$, $S_a$, $S_1$ 
we have a solution of $F(x,0)= 0$ of the form   
$x_j = 0$, if $j \in S_0$, 
$x_j = a$, if $j \in S_a$,
$x_j = 1$, if $j \in S_1$. 
Clearly, these are the only solutions of $F(x,0) = 0$.
An inspection of the Jacobian reveals that when $S_a$ is empty,  
these solutions are stable. 
On the other hand if $S_a$ is nonempty these solutions are unstable. 
The number of unstable direction is the number of sites in $S_a$.
The solutions where $ S_a $ is empty are generalizations of 
the fronts that exist for the one dimensional case, they are the
main subject of interest of the article.

\subsection{Homogeneous fixed points}

Let us now consider the case $\epsilon >0$. The homogeneous
fixed points can be analyzed for arbitrary
$\epsilon$. For that consider the system linearized around the
fixed point $x^*$
\be\label{lin}
{v}_t = [\epsilon \Delta   + Df(x^*)] v ,~\ee
where the Jacobian matrix has elements
\be\label{jacob}
Df(x^*)= \delta_{k,m}  \gamma (2 (1 + a) x^*_k - 3 {x^*_k}^2 - a).\ee
When the fixed points are homogeneous, $DN$ has a very simple
form, it can be written 
$$Df = -\gamma a I,~Df = \gamma (a-1) I,~ Df = \gamma a(1-a) I$$
respectively for $x^*=[0,\ldots,0]^T  ~x^*= [1,\ldots,1]^T,
i^*=[a, \ldots, a]^T$, where $I$ is the $N \times N$ identity.
The matrix $Df$ is then $c I $ for some real constant $c$, 
and $\sigma(\epsilon \Delta + Df(x^*))$ is $\sigma(\epsilon  \Delta) + c$. 
To study the stability it is then
convenient to use the basis of orthogonal eigenvectors 
of the symmetric matrix $\Delta$ \cite{crs01}
$$ \Delta V^k = -\omega_k^2 V^k ,$$
where the eigenfrequencies $\omega_k$ verify
$$\omega_1=0 \le \omega_2 \le \dots \le \omega_n  . $$
We write
\be\label{i_exp}
i = \alpha_1 V^1 + \alpha_2 V^2 \dots + \alpha_n V^n .\ee
Plugging the above expression into (\ref{jacob}) we get the
evolution of the amplitude 
\be\label{alpha_t}
{\dot \alpha_k} = -[ \epsilon \omega_k^2 + a ]{\alpha_k}
\ee
for the fixed point $x^*=[0,\ldots,0]^T$. Clearly it is stable for any
$\epsilon$. In a similar way we can show that 
$x^*=[1, \ldots,1]^T $ is always stable. The fixed point 
$x^*=[a, \ldots, a]^T$ 
is always unstable since 
we have an eigenvalue $- \epsilon \omega_1^2 + \gamma a(1 - a) > 0$.

\subsection{Non homogeneous fixed points}

For the non homogeneous fixed points the analysis
is not so simple. Let us first consider the case 
$\epsilon >0$ but small.
The implicit value theorem implies that each 
solution $x_0$ of $F(x,0) = 0$ can be continued uniquely, that is, it  
belongs to a unique  smooth one-parameter 
family of $x(\epsilon)$ satisfying $F(x(\epsilon),\epsilon) = 0$, 
$x(0) = x_0$, 
provided that $|\epsilon|$ is sufficiently small, see e.g. \cite{Z86}. 
The solution $x(\epsilon)$ of the local branch passing from 
$x(0)$ has the same stability as $x(0)$, for  
$|\epsilon|$ sufficiently small. 
This follows from the fact that all the solutions $x(0)$
are hyperbolic.

The numerical solutions below were obtained 
using the minpack implementation of Powell's hybrid 
Newton method \cite{P70}.
We start from $\epsilon = 0$, solving (\ref{stat-zel}) using Newton's method
and step in $\epsilon$. 
After some $\epsilon$, we continue stepping but use the pseudo-arc as 
a parameter\cite{K77} because we anticipate a fold. 
The linear stability of a solution $x(\epsilon_0)$
is computed readily by examining the eigenvalues of 
$D_1 F(x,\epsilon)$ at $x(\epsilon_0)$, $\epsilon_0$, i.e.  
\be 
(D_1 F(x,\epsilon))_{n,m} = \epsilon \Delta_{n,m} + 
\delta_{n,m}  \gamma  (2 (1 + a) x_n - 3 x_n^2 - a).   
\ee

We see numerically that all solutions
of $F(x,\epsilon) = 0$ with $\epsilon > 0 $ satisfy 
$x_j \in (0,1)$, forall $j \in \{1, \ldots, 6\}$.
This is also shown in Corollary \ref{physical-solutions} below. 
As we increase the value of $\epsilon$ along a branch of 
solutions continued from an $\epsilon = 0$ solution $x_0$, 
the linear stability remains unchanged, until some 
$\epsilon_0$, depending on the branch, where we see a fold. 
The branch is then continued by decreasing $\epsilon$, until we 
reach a different solution ${\tilde x}(0)$ of the $\epsilon = 0$
problem. After the fold the number of stable and 
stable eigenvalues changes. We observe that when $x(0)$
is stable, the branch changes stability 
at the fold, and ${\tilde x}(0)$ is unstable.
For example, setting $a = 0.1$, we see 
that the unstable $\epsilon = 0 $ solution 
$[1, 1, 1, 0, 0.1, 0]^T$ is connected to the stable  
$\epsilon = 0 $ solution $[1, 1, 1, 0, 0, 0]^T$
by a branch that has a fold at $ \epsilon = 0.00131035764$. 
In Fig. \ref{conti} we show the value of the component 
$x_5$ at different values of $\epsilon$ of the fixed point.
The other components start, and finish at the same values.

A similar behavior was observed for all the examples examined, 
except the spatially homogeneous solutions 
$c[1,\ldots,1]^T $ with $c = 1$, $a$, or $0$. From  relation
(\ref{fp-zel}) one can see that these exist for all $\epsilon$.
Based on our numerical observations
we conjecture that all 
$ 3^6 - 3$ inhomogeneous fixed points of the $ \epsilon = 0 $ problem
(we exclude the spatially homogeneous solutions)
belong to branches undergoing a fold bifurcation at some 
positive value of $\epsilon$, i.e. we have 
$ (3^6 - 3) /2$ branches with folds, connecting pairs of 
$\epsilon = 0 $ solutions.  
This conjecture can be checked numerically by continuing 
all $\epsilon = 0 $ fixed points. From the theoretical point of
view we can also show that non-spatially homogeneous fixed points 
cannot exist for arbitrarily large $\epsilon$.
We have

\begin{proposition}
\label{large-epsilon-equilibria}
There is an $\epsilon_c > 0$, such that all 
$(x,\epsilon)$, $x \in I^N$, $\epsilon > \epsilon_c$
that satisfy $F(x,\epsilon) = 0$ are of the form 
$x = c[1,\ldots,1]^T$, with $c = 0$, $a$, or $1$.
\end{proposition}

The proof is given at the end of this section. 
The dynamical importance of $\epsilon_c$ will be discussed further
in the next section. The general idea is that for 
$\epsilon > \epsilon_c$
all initial conditions ($\neq a[1,\ldots,1]^T$) should 
go to one of the two fixed points $c[1,\ldots,1]^T$, $c = 1$, $0$, 
as $t \rightarrow \infty$.
\begin{figure} [H]
\centerline{
\epsfig{file=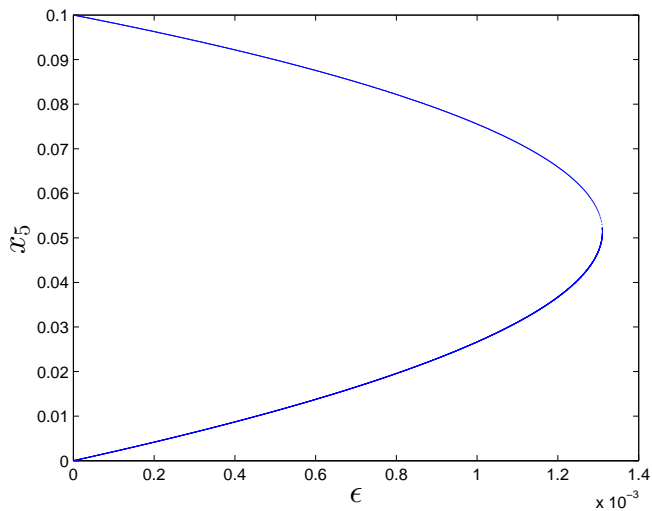,width=0.8\linewidth,angle=0}}
\caption{Component $x_5$ v.s. $\epsilon$ for a branch connecting the 
$\epsilon = 0$ fixed points
$(1, 1, 1, 0, 0.1, 0)$ and $(1, 1, 1, 0, 0, 0)$. }
\label{conti}
\end{figure}

An interesting problem is the computation of $\epsilon_c$. 
One idea is to continue all branches starting at $\epsilon = 0$ 
solutions and find the largest 
value $\epsilon_0$ of a fold. This computation would give 
a lower estimate of $\epsilon_c$, since we can not at present 
rule out the possibility of fixed points not belonging to these branches. 
Also it is of interest to see whether we can 
have a family of fixed points $x(\epsilon)$ that are stable 
for $\epsilon$ arbitrarily close to 
$\epsilon_c$, e.g. a continuous branch having a fold with 
change of stability at $\epsilon_c$.    
To obtain a first estimation of $\epsilon_c$ 
we have examined numerically all branches starting from 
stable $\epsilon = 0$ solutions for a fixed value of $a$.
There are $2^6 - 2$ such branches (we exclude 
$c[1,\ldots,1]^T $, with $c = 1$, $0$).
These are solutions $x(0)$
with $S_a = \emptyset$. In all (non-spatially homogeneous) 
cases these solutions are connected to 
an unstable solution ${\tilde x}(0)$ of the $\epsilon = 0 $ problem, 
with $S_a \neq  \emptyset$. 
For $a = 0.1$, 
the largest value of the 
fold coupling $\epsilon_0$ is 
${\overline \epsilon}_0 =  0.00299835224$,  
and is observed for the branch connecting
the $\epsilon = 0$ solutions $ [0, 0, 0, 0, 0, 1]^T  $ 
and $[0, 0, 0, 0, 0.1,  1]^T $.

Note that the $\epsilon = 0$ solution $ [0, 0, 0, 0, 0, 1]^T $ 
has only one neighbor. This is read from the Laplacian (\ref{lap_6}). 
It is reasonable to expect that the solutions that are 
the last to exist have the least neighbors. 
We see from (\ref{lap_6}) that all other $\epsilon = 0$ 
solutions with $|S_1| = 0$ have more that two neighbors, 
and it is observed that the corresponding branches undergo folds 
at smaller values of $\epsilon$. 
For example the branch starting from 
$ [0, 0, 0, 1, 0, 0]^T$, with two neighbors by (\ref{lap_6}), 
undergoes a fold
at $ \epsilon_0 =  0.00281313677$, 
while the branch starting from 
$[0, 0, 0, 0, 1, 0]^T$, with four neighbors, 
undergoes a fold
at $ \epsilon_0 = 0.00252927787$. 
The notion of neighbors can be extended to ($|S_a| = 0$) 
$\epsilon = 0$ solutions with $|S_1| > 1$. In such cases 
we can look for the number
of external connections to the set $S_1$, i.e. the number of 
points having distance one from $S_1$. 
We see that more sites in $S_1$ generally imply lower $\epsilon_0$
in the corresponding branch. For example 
the branch starting from 
$[1,1,1,1,1, 0]^T$, where $S_1$ has one external connection,  
undergoes a fold     
at  $ \epsilon_0  = 0.00250694432$. This is lower that the
value of the fold value $ \epsilon_0$ of the 
branch starting from 
$ [0, 0, 0, 1, 0, 0]^T$ above, with two neighbors but fewer peaks.
Comparing the values of $\epsilon_0$ for the
branches corresponding to $ [0, 0, 0, 0, 0, 1]^T $ 
and $[1,1,1,1,1, 0]^T$, 
we also see that complementary $\epsilon = 0$ solutions  
$x(0)$, $x'(0)$ (with $|S_a| = 0$), 
i.e. ones with $S_1(x(0)) $, $S_1(x'(0)) $ that are disjoint and 
whose union is the set of all nodes, generally have 
corresponding branches with different fold values. 

The $\epsilon = 0$ solutions not considered in the above enumeration
are expected to correspond to branches of solutions 
that are linearly unstable. Thus, even if we find 
a static solution such that 
$\epsilon_0 > {\overline \epsilon}_0=0.00299835224$,
we expect that for $\epsilon > {\overline \epsilon}_0$,  
almost all initial conditions of the time dependant system
(\ref{zeldovich}) go to either 
$c[1,\ldots,1]^T$, $c = 1$, $0$, 
as $t \rightarrow \infty$.

To better understand how $\epsilon_0$ depends
on the type of front and node connectivity, 
we develop a simple argument that assumes
$\epsilon_0$ is small, and that 
all sites except one that we call $n_c$ 
have values $1 + O(\epsilon)$, or $O(\epsilon)$, see 
subsection 3.3 below. This is consistent with what we see numerically, 
namely that the node that is destabilized first has value
approximatelly $a /2$, see e.g. Fig. \ref{conti}. 
Other sites have values that are much lower. 
The argument is as follows. Call $x$ the value of the node $n_c$ 
that will first destabilize.  Then the equation at $n_c$ for $x$ is
$$\epsilon ( N - K x + O(\epsilon) ) + x(1-x)(x-a) = 0~,$$
where $N$ is the number of neighbors of 
$n_c$ that are at 1 and $K$ is the connectivity
of $n_c$. This yields
\be\label{epsofx}
\epsilon = { x(1-x)(x-a) \over Kx -N} .\ee
From the continuation study of the static solutions we have seen that for
$$\epsilon=\epsilon_0,~~ x\approx a/2  ~.$$
Combining this observation with (\ref{epsofx}) yields the estimate for 
\be\label{epsofa}
\epsilon_0 = {a^2 \over 4} {2-a \over 2N-Ka}. \ee
This estimate is reported in Table \ref{tab1} 
below, together with the $\epsilon_o$
found numerically.
For $a$ relatively small, e.g. for $a= 0.1$ used here, 
we see excellent agreement.

\subsection{Asymptotics of the fixed points}

In what follows we show some general results on the profile 
of the fixed points of  
(\ref{zeldovich}) for $\epsilon > 0$, and small. 
We estimate the decay of the fixed point profiles 
away from the sites where the solution is near unity;
we also see that we can obtain small $ \epsilon$ 
asymptotics for $x(\epsilon)$ at all sites.  
For instance, we show that 
the amplitude $x_n(\epsilon)$ of the equilibrium 
at the site $n$ is 
$$ x_n(0) + O(\epsilon^{d_n}), $$
where $d_n$ is the distance of site $n$ from the 
analogue of the ``interface'' of the $\epsilon = 0$
configuration, see Lemma \ref{asympt-1}.  
Roughly speaking, the interface or ``front'' of an 
$\epsilon = 0$
configuration, defined more precisely below, 
consists of the sites where the solution jumps 
from zero to unity.   
The small $\epsilon $ asymptotic gives us
information on the decay of the $x_n(\epsilon)$ as we move
away form the sites that are near unity.
For sites with value near unity    
it also tells us that are 
further away from the interface have values 
that are much closer to unity. 
%
%Lemma \ref{asympt-1} can be used to obtain Proposition \ref{asympt-2}, 
%which is more precise and
%shows that for $\epsilon > 0$ (and small) 
%5all the solutions continued 
%from the $x(0)$ solutions belong to the cube $[0,1]^N$. 
%
 
Proposition \ref{asympt-2}
can be also used to compare 
small $\epsilon$ solutions continued from different $x(0)$, 
see Corollary \ref{solution-pairs} below.

The proof of Proposition \ref{asympt-2} is based on 
small$-\epsilon$ expansions
$$ x_n(\epsilon) = \sum^{\infty}_{m = 0} a_{n,m} \epsilon^m, $$ 
valid for all sites $n$. The idea is to insert these expression 
into (\ref{fp-zel}) and examine the coefficients of the series. 
We first obtain a less precise, intermediate statement,   
Lemma \ref{asympt-1}, using induction on the distance from the
``interface''  
between ones and zeros of the $\epsilon = 0$ 
solutions.   
Proposition \ref{asympt-2} uses the same strategy, 
and Lemma \ref{asympt-1}.

The precise statements use the following definitions and notation.

Let $\hbox{nbd}(n)$ denote the sites adjacent to the site $n$.
Let $c_n = | \hbox{nbd}(n)|$. 
Let $\hbox{dist}(R,n)$ denote the distance between 
the set of sites $R$, and a node $n$.  

Given a nontrivial solution $x(0)$ of the $\epsilon = 0$ 
equation $F = 0$,  
denote by $S_1$, $S_a$, $S_0$ the sets of indices $n$ where 
$x_n = 1$, $a$, $0$ respectively. Also let $S_A = S_1 \cup S_a$. 
Let $I$ be the set of nodes $n \in S_1$ having at least one neighbor
$j \in S_a \cup S_0$. 
The set $I$ plays the role of the ``interface'' of 
the configuration.\\
Then we have:

\begin{proposition}
\label{asympt-2}
Let $x(0)$ be a nontrivial solution of
equation (\ref{stat-zel}), (\ref{fp-zel}) 
with $\epsilon = 0 $, and 
let $x(\epsilon)$, $\epsilon \in [0,\epsilon_0]$ 
denote the unique branch of solutions of $F = 0$, $\epsilon > 0$, 
that continue $x(0)$ for $\epsilon \geq 0$.  
Consider the sets $S_1$, $S_a$, $S_0$, and $I$ corresponding to 
$x(0)$ as defined above, with $S_i$, $I$ nonempty.
Then for $\epsilon > 0$ sufficiently small we have that 
(i) $n \in S_0$, $\hbox{dist}(S_A,n) = m \geq 1 $ imply    
\begin{equation}
\label{0-site-asympt-2}
x_n(\epsilon) = a_{n,m} \epsilon^{m} + O(\epsilon^{m+1}), \quad\hbox{with}\quad 
a_{n,m}  > 0, 
\end{equation}
and (ii) $n \in S_1$, $ \hbox{dist}(I,n) = m \geq 0$ imply
\begin{equation}
\label{1-site-asympt-2}
x_n(\epsilon) = 1+ a_{n,m+1}\epsilon^{m + 1} + O(\epsilon^{m+2}),  
\quad\hbox{with}\quad 
a_{n,m+1}  <  0,
\end{equation}
\end{proposition}

An immediate consequence is: 

\begin{corollary}
\label{physical-solutions}
Let $x(0)$ be a nontrivial solution
equation (\ref{fp-zel}) with $\epsilon = 0 $, and 
let $x(\epsilon)$, $\epsilon \in [0,\epsilon_0]$ 
denote the unique branch of solutions of (\ref{fp-zel}), 
$\epsilon > 0$, 
that continue $x(0)$ for $\epsilon \geq 0$.  
Then for $\epsilon > 0$ and sufficiently small
we have $x_n(\epsilon) \in (0,1)$, for all sites $n$.  
\end{corollary}

\begin{proof}
For sites $n \in S_a$ we have 
$x_n(\epsilon) = a + O(\epsilon) \in (0,1)$ for $\epsilon $
sufficiently small. 
For other sites the statement follows 
form Proposition \ref{asympt-2}.
\end{proof}

\begin{remark}
The above asymptotic is appears to be related to the 
estimate of $\epsilon_0$ in \ref{epsofa}, 
and the
assumption that the all sites $n \neq n_c$ have 
values $1 + O(\epsilon)$, and $O(\epsilon)$.
Indeed most sites $n \neq n_c$ are seen to be $O(\epsilon)$ from their 
$\epsilon = 0$ values at $\epsilon_0$.  
Note however that the site $n_c$ also has the value $0$ (or $\alpha$)
at $\epsilon = 0$, and comes near $ a / 2 $ as $\epsilon$ approaches 
$\epsilon_0$. The use of the small$-\epsilon$ 
asymptotic in justifying  \ref{epsofa} is not clear. 
\end{remark}

Another consequence of 
Proposition \ref{asympt-2}
is that 
for $\epsilon > 0$ sufficiently small 
there exist pairs of static solutions 
$x$, $y$ of the Zeldovich 
equation satisfying $x_n < y_n $, 
$\forall n \in \{1, \ldots, N\}$.
The construction is as follows:

\begin{corollary}
\label{solution-pairs}
Let $x(\epsilon)$, $y(\epsilon)$, $\epsilon$ 
sufficiently small, be continuations
of the $\epsilon = 0$ fixed points $x = x(0)$, 
$y = y(0)$ of the Zeldovich equation satisfying
\begin{equation}
(i)     \quad S_\alpha(x) = S_\alpha(y) = \emptyset, \quad
(ii)    \quad S_1(x) \subset S_1(y), 
\end{equation}
\begin{equation}
(iii) \quad \hbox{dist}(I(x),n) < \hbox{dist}(I(y),n), \quad 
\forall n \in S_1(x) \cup S_1(y),
\end{equation}
\begin{equation}
(iv) \quad \hbox{dist}(S_1(x),n) > \hbox{dist}(S_1(y),n), \quad \forall
n \in S_0(x) \cup S_0(y).
\end{equation}
Then 
for all $\epsilon > 0$ sufficiently small we have 
$x_n(\epsilon) < y_n(\epsilon)$, 
$\forall n \in \{1, \ldots, N\}$. 
\end{corollary}

\begin{proof}
We consider the three cases $n \in S_0(x) \cup S_1(y)$, 
$S_1(x) \cup S_1(y) $, and $S_0(x) \cup S_0(y) $. 
By (ii) $S_1(x) \cup S_0(y) = \emptyset $.
For  $n \in S_0(x) \cup S_1(y)$ we have 
\begin{equation}
y_n(\epsilon) = 1 - O(\epsilon) > x_n(\epsilon) = O(\epsilon),
\notag 
\end{equation}
for $\epsilon > 0$ small.

For $n \in S_1(x) \cup S_1(y)$,  
Proposition \ref{asympt-2} yields
\begin{eqnarray}
x_n(\epsilon) = 1 - |a_{n,m+1}| \epsilon^{m+1} + O(\epsilon^{m+2}), 
\notag \\
y_n(\epsilon) = 1 - |a_{n,{\tilde m}+1}| \epsilon^{{\tilde m}+2} 
+ O(\epsilon^{{\tilde m}+2}), \notag
\end{eqnarray}
with $a_{n,m}$, $ a_{n,{\tilde m}+1} \neq 0$, and 
\begin{equation}
m = \hbox{dist}(I(x),n), \quad 
{\tilde m} =  \hbox{dist}(I(y),n), \quad {\tilde m} > m. 
\notag
\end{equation} 
Therefore $ y_n(\epsilon) > x_n(\epsilon) $ for $\epsilon > 0 $
small enough.

For $n \in S_0(x) \cup S_0(y)$, 
Proposition \ref{asympt-2} yields
\begin{eqnarray}
x_n(\epsilon) = |a_{n,\mu}| \epsilon^{\mu} + O(\epsilon^{\mu + 1}), 
\notag \\
y_n(\epsilon) = |a_{n,{\tilde \mu}}| \epsilon^{{\tilde \mu}} 
+ O(\epsilon^{{\tilde \mu}+1}), \notag
\end{eqnarray}
with $a_{n,\mu}$, $ a_{n,{\tilde \mu}+1} \neq 0$, and
\begin{equation}
\mu = \hbox{dist}(I(x),n), \quad 
{\tilde \mu} =  \hbox{dist}(I(y),n), \quad \mu > {\tilde \mu}. 
\notag
\end{equation}
Again $ y_n(\epsilon) > x_n(\epsilon) $ for $\epsilon > 0 $
small enough.
\end{proof}

The proof of Proposition \ref{asympt-2}
uses the following intermediate result. 

\begin{lemma}
\label{asympt-1}
Let $x(0)$ be a nontrivial solution
equation $F = 0$ with $\epsilon = 0 $, and 
let $x(\epsilon)$, $\epsilon \in [0,\epsilon_0]$ 
denote the unique branch of solutions of $F = 0$, $\epsilon > 0$, 
that continue $x(0)$ for $\epsilon \geq 0$.  
Consider the sets $S_1$, $S_a$, $S_0$, and $I$ corresponding to 
$x(0)$ as defined above, with $S_1$, $I$ nonempty.
Then for $\epsilon > 0$ sufficiently small we have that 
(i) $n \in S_0$, $\hbox{dist}(S_A,n) \geq  m \geq 1 $ imply    
\begin{equation}
\label{0-site-asympt-1}
x_n(\epsilon) = O(\epsilon^{m}), 
\end{equation}
and (ii) $n \in S_1$, $ \hbox{dist}(I,n) \geq m \geq 0$ imply
\begin{equation}
\label{1-site-asympt-1}
x_n(\epsilon) = 1+ O(\epsilon^{m + 1}).
\end{equation}
\end{lemma}

\begin{proof}
We use the analytic version of the implicit value theorem, 
which allows us to write $x_n(\epsilon)$ as a convergent power series in
$\epsilon$,
for $\epsilon$ sufficiently near the origin. 
Thus we 
write $x_n(\epsilon) = \sum^{\infty}_{m = 0} a_{n,m} \epsilon^m$, 
for all sites $n$, see e.g. \cite{Z86}.
(Since the network is finite it is sufficient to use the $C^r$
version for $r$ sufficiently large.) 

We then already have $x_n(\epsilon) = O(\epsilon)$, $ \forall n  \in S_0$, 
and $x_n(\epsilon) = 1 + O(\epsilon)$, $ \forall n \in S_1$.

To show (i) let $n$ satisfy
$\hbox{dist}(S_A,n) \geq 2$. We have 
\begin{eqnarray}
\label{a1-0-1-eq1}
\epsilon (\Delta x)_n  & = & \epsilon [ -c_n(a_{n,1} \epsilon + O(\epsilon^2)) + 
\sum_{j \in \hbox{nbd}(n)} x_j ]   \notag \\
& = & O(\epsilon^2), 
\end{eqnarray}
since $j \in \hbox{nbd}(n)$ implies $j \in S_0$, hence $x_j = O(\epsilon)$. 

Also
\begin{equation}
\label{a1-0-1-eq2}
x_n(1-x_n)(a - x_n) =  - a a_{n,1} \epsilon + O(\epsilon^2).   
\end{equation}
By (\ref{a1-0-1-eq1}), (\ref{a1-0-1-eq2}), and $F = 0$ 
we must then have $a_{n,1} = 0$. 

We use induction:  
suppose that 
if $\hbox{dist}(S_A,n) \geq m \geq 2$, then $x_n = O(\epsilon^m)$.

Then for $n$ satisfying $\hbox{dist}(S_A,n) \geq m + 1$ we have 
\begin{eqnarray}
\label{a1-0-m-eq1}
\epsilon (\Delta x)_n  & = & \epsilon [-c_n(a_{n,m} \epsilon^m +
O(\epsilon^{m+1})) + 
\sum_{j \in \hbox{nbd}(n)} x_j ] \notag \\
& = & O(\epsilon^{m+1}), 
\end{eqnarray}
since $j \in \hbox{nbd(n)}$ implies $\hbox{dist}(S_A,n) \geq m$, 
hence $x_j = O(\epsilon^m)$ by the inductive hypothesis. 
On the other hand
\begin{equation}
\label{a1-0-m-eq2}
x_n(1-x_n)(a - x_n) =  - a a_{n,m} \epsilon^m + O(\epsilon^{m+1}).   
\end{equation}
By (\ref{a1-0-m-eq1}), (\ref{a1-0-m-eq2}), and $F = 0$  
we must then have $a_{n,m} = 0$, 
and therefore $x_n = O(\epsilon^{m+1})$, as required. 

To see (ii) let $n \in S_1$ satisfy
$\hbox{dist}(I,n) = 1$, 
so that all $j \in \hbox{nbd}(n)$ satisfy $x_j(0) = 1 $.
Also $x_n = 1 + O(\epsilon)$.
Then 
\begin{eqnarray}
\label{a1-1-1-eq1}
\epsilon (\Delta x)_n  & = & \epsilon [ -c_n (1 + a_{n,1} \epsilon +
O(\epsilon^2)) +
\sum_{j \in \hbox{nbd}(n)} x_j ] \notag \\
&  = &  \epsilon [ -c_n  - c_n a_{n,1} \epsilon  +
c_n + \sum_{j \in \hbox{nbd}(n)} a_{j,1} \epsilon  + O(\epsilon^2)] \notag \\
& = & O(\epsilon^2).
\end{eqnarray} 
On the other hand 
\begin{equation}
\label{a1-1-1-eq2}
x_n(1-x_n)(a - x_n) =  - (a - 1) a_{n,1} \epsilon + O(\epsilon^2).   
\end{equation}
By (\ref{a1-1-1-eq1}), (\ref{a1-1-1-eq2}), and $F = 0$ we must 
have $a_{n,1} = 0$, 
and therefore $x_n = O(\epsilon^2)$.

For the inductive step, assume that if $n \in S_1$ 
satisfies $ \hbox{dist}(I,n) \geq m$, then $x_n = 1 + O(\epsilon^{m+1})$. 
Consider then a site $n$ satisfying $ \hbox{dist}(I,n)  \geq   m + 1$,
then 
\begin{eqnarray}
\label{a1-1-m-eq1}
\epsilon (\Delta x)_n  & = & \epsilon [ -c_n (1 + a_{n,m+1} \epsilon^{m+1} +
O(\epsilon^{m 
+2})) +
\sum_{j \in \hbox{nbd}(n)} x_j ] \notag \\
&  = &  \epsilon [ -c_n  - c_n a_{n,m+1} \epsilon^{m+1}  +
c_n + \sum_{j \in \hbox{nbd}(n)} a_{j,m} \epsilon^{m+1}  + O(\epsilon^{m+2})] 
\notag \\
& = & O(\epsilon^{m+2}), 
\end{eqnarray} 
using the fact that $j \in \hbox{nbd}(n)$ implies 
$ \hbox{dist}(I,j) \geq m$,
hence $  x_j = 1 + O(\epsilon^{m+1})$ by the inductive hypothesis.
On the other hand 
\begin{equation}
\label{a1-1-m-eq2}
x_n(1-x_n)(a - x_n) = -(a - 1) a_{n,m+1} \epsilon^{m+1} + O(\epsilon^{m + 2}).   
\end{equation}
By (\ref{a1-1-m-eq1}), (\ref{a1-1-m-eq2}), $F = 0$ implies $a_{n,m+1} = 0$, 
and therefore $x_n = 1 + O(\epsilon^{m+2})$, as required.  
\end{proof}

%\begin{lemma}
%\label{asympt-2}
%Let $x(0)$ be a nontrivial solution
%equation (\ref{}) with $\epsilon = 0 $, and 
%let $x(\epsilon)$, $\epsilon \in [0,\epsilon_0]$ 
%denote the unique branch of solutions of $F = 0$, $\epsilon > 0$, 
%that continue $x(0)$ for $\epsilon \geq 0$.  
%Consider the sets $S_1$, $S_a$, $S_0$, and $I$ corresponding to 
%$x(0)$ as defined above, with $S_i$, $I$ nonempty.
%Then for $\epsilon > 0$ sufficiently small we have that 
%(i) $n \in S_0$, $\hbox{dist}(S_A,n) = m \geq 1 $ imply    
%\begin{equation}
%\label{0-site-asympt-2}
%x_n(\epsilon) = a_{n,m} \epsilon^{m} + 
%O(\epsilon^{m+1}), \quad\hbox{with}\quad 
%a_{n,m}  > 0, 
%\end{equation}
%and (ii) $n \in S_1$, $ \hbox{dist}(I,n) = m \leq 0$ imply
%\begin{equation}
%\label{1-site-asympt-2}
%x_n(\epsilon) = 1+ a_{n,m+1}\epsilon^{m + 1} + O(\epsilon^{m+2}),  
%\quad\hbox{with}\quad 
%a_{n,m+1}  <  0,
%\end{equation}
%\end{lemma}

We now prove Proposition \ref{asympt-2}.

\begin{proof}
The starting point is again 
the expression $x_n(\epsilon) = \sum^{\infty}_{m = 0} a_{n,m} \epsilon^m$.
To see (i) first consider sites $n$ satisfying 
$\hbox{dist}(S_A,n) = 1$. 
Letting $J_1$ be the set of sites 
$j \in \hbox{nbd}(n) \cap S_A$, and  
$J_2 = \hbox{nbd}(n) \setminus J_1$, 
we have 
$$ |J_1| = | \hbox{nbd}(n) \cap S_1 | + 
 |  \hbox{nbd}(n) \cap S_a |  >  0. $$
Then 
\begin{eqnarray}
\label{0-1-eq1}
\epsilon (\Delta x(\epsilon))_n  & =  & 
\epsilon [ -c_n O(\epsilon) + \sum_{j \in J_1} x_j + \sum_{j \in J_2} x_j ] 
\notag   \\
&  = &  (  | \hbox{nbd}(n) \cap S_1 |  
+ |  \hbox{nbd}(n) \cap S_a |  a) \epsilon
+ O(\epsilon^2) \notag  \\
&  > &  0
\end{eqnarray}
for $\epsilon > 0$ sufficiently small.
On the other hand 
\begin{equation}
\label{0-1-eq2}
x_n (1 - x_n) (x_n - a)  =  - a \epsilon a_{n,1}  + O(\epsilon^2).  
\end{equation}
By (\ref{0-1-eq1}), (\ref{0-1-eq2}), $F = 0$, 
we need $a_{n,1} > 0$. 

We proceed inductively, assuming that 
if $n \in S_0$ satisfies $\hbox{dist}(S_A,n) = m \geq 1$, then  
$x_n(\epsilon) = a_{n,m} \epsilon^m + O(\epsilon^{m+1})$, with $a_{n,m} > 0$.
Consider then a site $n$ satisfying  $\hbox{dist}(S_A,n) = m +  1$. 
Let $J_m $ be the set of sites $j \in \hbox{nbd}(n)$ satisfying
$\hbox{dist}(S_A,n) = m$. 
Clearly $|J_m| > 0$. 
Also let $J_{m+1} = \hbox{nbd}(n) \setminus J_m$.
By Lemma \ref{asympt-1}, $x_n(\epsilon) = O(\epsilon^{m+1})$.  
Then
\begin{eqnarray}
\label{0-m-eq1}
\epsilon (\Delta x(\epsilon))_n  & =  & 
\epsilon [ -c_n (a_{n,m+1} \epsilon^{m+1} + O(\epsilon^{m + 2}))  + \sum_{j \in
J_m} x_j + 
\sum_{j \in J_{m+1}} x_j ]   \notag \\
&  = &     \sum_{j \in J_m} a_{j,m} \epsilon^{m+1} +  O(\epsilon^{m+2}) 
\notag \\
&  > &  0
\end{eqnarray}
for $\epsilon > 0$ small, since $a_{j,m} > 0$, $\forall j \in J_m$,
by the inductive hypothesis.
On the other hand 
\begin{equation}
\label{0-m-eq2}
x_n (1 - x_n) (x_n - a)  =  - a a_{n,m+1}  \epsilon^{m+1}+ O(\epsilon^{m+ 2}).  
\end{equation}
By (\ref{0-m-eq1}), (\ref{0-m-eq2}), and $F = 0$,  
we therefore need $a_{n,m+1} > 0$. 
 
To see (ii) consider a site 
$n \in S_1 \cap I$, so that $\hbox{dist}(I,n) = 0$. 
Then $x_n(\epsilon) = 1 + a_{n,1} \epsilon + O(\epsilon^2)$, 
and 
\begin{eqnarray}
\label{1-1-eq1-1}
\epsilon (\Delta x(\epsilon))_n  & =  &
\epsilon [   -c_n(1 + O(\epsilon)) + \sum_{j \in \hbox{nbd}(n) \cap S_1} x_j + 
\sum_{j \in \hbox{nbd}(n) \cap S_a} x_j ]  \notag \\
& = & \epsilon (-c_n + |\hbox{nbd}(n) \cap S_1| + 
|\hbox{nbd}(n) \cap S_\alpha| a) 
+ O(\epsilon^2) ].  
\end{eqnarray}

Suppose 
$\mu = |\hbox{nbd}(n) \cap S_\alpha|  \geq 1$, 
then $|\hbox{nbd}(n) \cap S_1| \leq c_n - \mu $, 
and (\ref{1-1-eq1-1}) yield 
\begin{eqnarray}
\label{1-1-eq1-2}
\epsilon (\Delta x(\epsilon))_n  & \leq & 
\epsilon (-c_n + (c_n - \mu) + a \mu) +
O(\epsilon^2) \notag \\
& = &  (-1 + a) \mu \epsilon + O(\epsilon^2) \notag  \\
& < & 0, 
\end{eqnarray}
for $\epsilon > 0$ sufficiently small. 
If $\mu = 0$, $n \in I$ implies 
$ |\hbox{nbd}(n) \cap S_1| <  c_n $, 
so that  (\ref{1-1-eq1-1})
implies 
\begin{eqnarray}
\label{1-1-eq1-3}
\epsilon (\Delta x(\epsilon))_n  & \leq & \epsilon (-c_n + (c_n - 1) +
O(\epsilon^2) \notag \\
& = &  - \epsilon + O(\epsilon^2) \notag \\
& < & 0, 
\end{eqnarray}
for $\epsilon > 0$ sufficiently small.
Combining  (\ref{1-1-eq1-2}), (\ref{1-1-eq1-3}) with 
\begin{equation}
\label{1-1-eq2}
x_n (1 - x_n) (x_n - a)  =  - a_{n,1}  (1 - \alpha) \epsilon+ 
O(\epsilon^{2}),  
\end{equation}
we see that to satisfy $F = 0$ with $\epsilon > 0$, 
sufficiently small we must have 
$  - a_{n,1} >  0 $.

For the inductive step, assume that  
$n \in S_1 $, $\hbox{dist}(I,n) = m$ imply 
$x_n(\epsilon) = 1 + a_{n,m+1} \epsilon^{m+1} + O(\epsilon^{m+2})$
with $a_{n,m+1} < 0$. 
Then let $n \in S_1$, $\hbox{dist}(I,n) = m + 1$.
Let $J_m$ be the set of sites $j \in \hbox{nbd}(n)$ 
satisfying $\hbox{dist}(I,n) = m$,
let $J_{m+1} $ be the set of sites $j \in \hbox{nbd}(n)$ 
satisfying $\hbox{dist}(I,n) \geq m + 1$. 

By Lemma \ref{asympt-1} 
we have $x_n(\epsilon) = 1 + O(\epsilon^{m+2})$. 
Then  
\begin{eqnarray}
\label{1-m-eq1}
\epsilon (\Delta x(\epsilon))_n  & = &
\epsilon [-c_n(1 + O(\epsilon^{m+2})) + \sum_{j \in J_m} x_j + 
 \sum_{ j \in J_{m+1} } x_j] \notag \\
 & = &  \epsilon [-c_n  + |J_m| + 
\sum_{j \in J_m}  a_{j,m+1} \epsilon^{m+1} + 
(c_n - |J_m|) + O(\epsilon^{m+2}) ] \notag \\
& = &  \sum_{j \in J_m}  a_{j,m+1} \epsilon^{m+2} 
+ O(\epsilon^{m+3}) ] \notag \\
& < & 0
\end{eqnarray}
for $\epsilon > 0$ sufficiently small, 
since $a_{j,m+1} < 0$, $\forall j \in J_m$
by the inductive hypothesis.
On the other hand, 
\begin{equation}
\label{1-1-eq2a}
x_n (1 - x_n) (x_n - a)  =  - (1 -a) a_{n,m+2}  \epsilon^{m+2} + 
O(\epsilon^{m+3}).
\end{equation}
By (\ref{1-m-eq1}), (\ref{1-1-eq2a}) to satisfy $F = 0$ 
we must have 
$- a_{n,m+2} > 0 $, as required.  
\end{proof}

We now prove Proposition \ref{large-epsilon-equilibria}

\begin{proof}
To study large $\epsilon$ solutions of 
$F(x,\epsilon) = 0$
%where 
%\begin{equation}  
%\label{Fp-zel-1}
%{F}(x,\epsilon) =  
%\epsilon (\Delta x)_n  + f_n(x), \quad f_n(x)\gamma (1-x_n) x_n (x_n-a), 
%\end{equation}
%$n = 1, \ldots, N$, 
we will equivalently examine 
$\mu \rightarrow 0^+$ solutions 
${\tilde F} (x,\mu) = 0$, where
\begin{equation}  
\label{Fp-zel-mu}
{\tilde F}_n(x,\mu) =  
(\Delta x)_n  + \mu f_n(x), 
\end{equation}
$n = 1, \ldots, N$. 

Then $F(x,\epsilon) = 0$, $\epsilon > 0$, is equivalent 
to ${\tilde F}(x,\mu) = 0$, with $ \mu = \epsilon^{-1}$.

Consider a a sequence 
$\{(x_n,\epsilon_n)\}_{n \in {\bf Z}^+} \in I^N \times {\bf R}^+$, 
satisfying $\epsilon_n \rightarrow \infty$, and 
$F(x_n,\epsilon_n) = 0$, $\forall n \in {\bf Z}^+$.
Such sequences clearly exist. 
Moreover 
$(x_n,\mu_n)$, with $\mu_n = (\epsilon_n)^{-1}$, satisfy 
${\tilde F}(x_n,\mu_n) = 0$, $\forall n \in {\bf Z}^N$.
The sequence of solutions  
$\{ (x_n,\mu_n) \}_{n > n_0} $ of ${\tilde F} = 0$
belongs to $ I^{N+1} =I^{N} \times [0,1]$ for some $n_0 > 0$, 
and by the compactness of $ I^{N+1}$ 
has a convergent subsequence in $ I^{N+1}$, 
denoted again as $\{(x_n,\mu_n)\}_{n \in {\bf Z}^+} $.
Let $(x_*,\mu_*)$ be the limit of this subsequence.
By the assumption $\epsilon_n \rightarrow \infty$, we have 
that $\mu_* = 0$. 
Also, ${\tilde F}:{\bf R}^{N+1} \rightarrow {\bf R}^{N+1}$ 
is continuous and therefore 
${\tilde F}(x_n,\mu_n) \rightarrow {\tilde F}(x_*,0)$
as $(x_n,\mu_n) \rightarrow (x_*,0)$. Therefore 
${\tilde F}(x_*,0) = 0$. Since ${\tilde F}(x,0) = \Delta x$ we 
have $x_* \in V \cap I^N $, where 
$V = \{c [1, \ldots, 1]^T \in {\bf R}^N: c \in {\bf R} \}$,
i.e. the kernel of $\Delta$. 

We show that $x_*$ can only be one of the 
$c[1,\ldots,1]^T$, with $c = 0$, $a$, or $1$.
Let $P $ the orthogonal projection of ${\bf R}^N$ onto $V$. 
Also let $W = I - P$, where $I$ the identity in ${\bf R}^N$. 
We apply $P$ and $ I - P$ to ${\tilde F} = 0$, and 
write $x = v + w$, 
with $v \in V$, $w \in W$. This decomposition is unique.
Using the facts that 
$\Delta$ and $P$ commute, and that $\Delta v = 0$,   
${\tilde F} = 0$ becomes 
\begin{equation}
\label{decomp-1}
  Pf(v+w) = 0, 
\end{equation}
\begin{equation}
\label{decomp-2}
\Delta w + \mu (I - P)f(v + w) = 0.
\end{equation}
Fix any $v \in V \cap I^N$. 
We use the implicit function theorem to continue  
the solution $(w,\mu) = (0,0)$ of 
(\ref{decomp-2}) to a solution with $\mu \neq 0 $. 
Then for $|\mu|$ sufficiently small there exists 
a one-parameter family of solutions $(w,\mu) = (h(\mu;v),\mu)$
of (\ref{decomp-2}), 
where $h(\cdot;v)$
is continuous in $\mu$, 
with $h(\mu;v)= O(\mu)$ as $\mu \rightarrow 0$ (uniformly in $v$).
The implicit function theorem also implies that these 
solutions are the only solutions 
of (\ref{decomp-2}) in a sufficiently small 
neighborhood of $(w,\mu) = (0,0)$ in 
$W \times {\bf R}$.
Similar considerations show that  
the function $h$ is continuous in $v$, 
$\forall v \in v \in V \cap I^N$. 

Thus all solutions of ${\tilde F}(x,\mu)$, 
with $x = v + w$, $v \in V$, $w \in W$, and
$w \rightarrow 0$, $\mu \rightarrow 0$, 
must be of the form 
$x = v + h(\mu,v)$, with $v$ a solution of 
\begin{equation}
\label{decomp-3}
  g(v,\mu) = Pf(v+h(\mu;v)) = 0, 
\end{equation}
by (\ref{decomp-1}).

Suppose that we have a sequence of solutions
$\{(v_n,\mu_n)\}_{n \in {\bf Z}^+}$ of $g(v,\mu) = 0$ with 
$v \in V \cap I^N $, and $ \mu_n \rightarrow 0$. 
By compactness
this sequence has a convergent subsequence.  
Denote its limit by $(v_*,0)$. By the
continuity of $h$, and therefore of $g$, $v_*$ must satisfy 
$g(v_*,0) =  Pf(v_*) = 0$, 
hence $v_* = v_r = c_r[1, \ldots, 1]^T$, 
$r = 1, 2, 3$ with 
$c_1 = 0$, $c_2 = a$, $c_3 = 1$. 
Applying the implicit function theorem again 
we check that each of the 
solutions $v_r$, $r = 1,2,3$, of $g(v,0) = 0$ is  
continued to a unique branch of solutions of 
$g(v,\mu) = 0$, with $(v,\mu)$. Each of these three
branches contains all possible solutions of $g= 0$ 
sufficiently near the 
respective $(v_r,0)$, $r = 1,2,3$.
By uniqueness   
these three local branches are subsets of the three trivial
branches $(v_r,\mu)$, $r = 1, 2, 3$, $\mu > 0 $, 
of solutions of ${\tilde F} = 0$.   
\end{proof}

\section{Comparison results for front propagation models}

We now consider the time dependant solutions of 
the Zeldovich equation (\ref{zeldovich}), and establish
qualitative comparison (or monotonicity) results for 
different solutions of the Zeldovich model, see Proposition 
\ref{compare-Zel-Zel}. An application is 
Corolary \ref{compare-static-sol}, a stability statement 
for some of the static solutions 
discussed in Corolary \ref{solution-pairs} 
of the previous section.
Another goal is to compare the Zeldovich model with the original 
Kermack-McKendrick system, and the intermediate
Fisher system.   
We show that the Fisher model describes a faster propagation
of the epidemic than both the Zeldovich and Kermack-McKendrick models, 
see Propositions \ref{compare-Zel-Fis}, \ref{compare-Fis-KmK} 
respectively.
In the next section we show some numerical examples.

The comparison statements below use a notion of ``partial 
order'' between configurations. In particular  
$ u < n$ (respectively $u \leq v$), with 
$u$, $v \in [0,1]^N$, 
will mean $u_n < v_n$ 
(respectively $u_n \leq v_n$), 
$\forall n \in \{1, \ldots, N\}$.
We also let ${\bf 0} = [0,\ldots,0]^T$, ${\bf 1} = [1, \ldots, 1]^T \in
[0,1]^N $. 
A ``larger'' configuration thus describes a state where the epidemic 
is more advanced at all sites.

\begin{proposition} 
\label{compare-Zel-Zel}
Let $T> 0$, $x$, $y: [0,T] \rightarrow [0,1]^N$ 
be two solutions of (either) the Zeldovich (\ref{zeldovich}) 
(or the Fisher (\ref{fisher})) equation, 
with initial conditions satisfying 
$ {\bf 0} <  x(0) < y(0) \leq {\bf 1}$. 
Then $x(t) < y(t)$, $\forall t \in [0,T]$. 
\end{proposition}

Since both vectors 
${\bf 0}$, ${\bf 1}$ are static solutions of the 
Zeldovich and Fisher equations, Lemma \ref{invariance-of-unit-cube}
is a special case of
Proposition \ref{compare-Zel-Zel}.

\begin{corollary}
\label{compare-static-sol}
Let $T > 0$.
Let $x$, $y \in [0,1]^N$ 
be two static solutions of the Zeldovich 
equation satisfying 
$x < y$, and let $u: [0,T] \rightarrow [0,1]^N$ 
be a solution
of the Zeldovich equation with initial condition $u(0)$ 
satisfying $ x <  u(0) < y$. Then 
$x < u(t) < y$, $\forall t \in [0,T]$.
\end{corollary}

The existence of pairs of 
static solutions of the Zeldovich equation 
satisfying $x < y$, is 
shown in the previous 
section, in Corollary \ref{solution-pairs}.

We now compare solutions of the Fisher and Zeldovich
equations. 

\begin{proposition}
\label{compare-Zel-Fis}
Let $T > 0$. Let $x_F$, $x_Z:[0,T] \rightarrow [0,T]^N$
be solutions of the Fisher (\ref{fisher}), 
and Zeldovich (\ref{zeldovich}) equations
respectively, with corresponding initial conditions satisfying 
$ {\bf 0} <  x_Z(0) \leq x_F(0) \leq {\bf 1}$. 
Then $x_Z(t) \leq   x_F(t)$, $\forall t \in (0,T]$.
\end{proposition}

\begin{proposition}
\label{compare-Fis-KmK}
Let $T > 0$. Let $x_F:[0,T] \rightarrow [0,1]^N$ be a solution
of the Fisher equation (\ref{fisher}),
and let $(s, i):[0,T] \rightarrow \tau^N$ 
satisfy the Kendrick-McKormack
system (\ref{km3}).
Suppose also that the corresponding 
initial conditions satisfy $i(0) <  x_F(0)$. 
Then $i(t) \leq x_F(t)$, $\forall t \in [0,T]$. 
\end{proposition}

The above comparison 
statements follow from analogous statements
for discrete time approximations of the solutions of 
the three equations. The approximations we use  
are obtained by the first-order explicit 
Euler method. 

Consider a general ODE 
${\dot z} = F(z)$ in ${\bf R}^K$ with initial
condition $z(0)$, and the corresponding solution 
$z$ in the interval $[0,T]$.
Fix a positive integer $M > 1$, and let  
$\Delta t = T / M$. 
Let $z^M$ be an array of $M+1$ vectors 
$z^M(m \Delta t) \in {\bf R}^K$, 
$m \in \{0, \ldots, M \}$, defined iteratively by
$z^M(0) = z(0)  $,  
\begin{equation}
\label{Euler-iterate}
z^M((m+1) \Delta t) = z^M(m \Delta t) + F(z^M(m \Delta t)) \Delta t,
\quad m \in {0, \ldots, M}.
\end{equation}
(The dependence of $\Delta t$ on $M$ is not explicit in this notation.)
Thus $z^M$ is the numerical trajectory 
obtained by the first order, explicit Euler method
with constant time-step 
${\Delta t} = {T/M} $ over an interval $[0,T]$.
We recall a standard convergence result 
for the Euler method (see e.g. \cite{I08}):

\begin{lemma} 
\label{Euler-meth-lem}
Consider the ODE 
$ {\dot z} = F(z)$ in ${\bf R}^K$ with initial
condition $z(0)$, and assume that the solution
$z(t)$, $t \in [0,T]$, exists and is unique. 
Assume also that $F$ is $C^1$ in ${\bf R}^K$. 
For every integer $M > 1$ let $z^M$ be as in 
(\ref{Euler-iterate}) with fixed  
time-step $ {\Delta t} = {T / M}$ over the interval $[0,T]$.
Then 
\begin{equation}
\label{Euler-conv}
\lim_{M \rightarrow \infty} 
\max_{m \in \{0, \ldots, M\} } 
||z(m {\Delta t}) -  z^M( m {\Delta t})|| = 0,
\end{equation}
where $|| \cdot ||$ denotes the norm in 
${\bf R}^K$.
\end{lemma}

In the following lemma we compare 
two Euler approximates of either the 
Zeldovich or the Fisher equations. 
We see that they preserve the order of the
initial conditions. 

\begin{lemma}
\label{comparison-Zel-Euler}
Let $T > 0$, and let $M > 1$ be an integer. 
Let $x^M$, $y^M$ be the 
Euler approximations with
time-step ${\Delta t} = T /M $ over the interval $[0,T]$ 
of two trajectories of either the Zeldovich (\ref{zeldovich})
or the Fisher (\ref{fisher}) equations.
Assume that $x^M(0)  < y^M(0)$, 
with $ x^M(0)$, $y^M(0) $ in $[0,1]^N$.
Then 
$x^M(m {\Delta t}) < y^M(m {\Delta t})$, 
$\forall m \in \{1, \ldots, M \}$, 
provided $M$ is sufficiently large.   
\end{lemma}

The lemma also implies that 
the Euler approximations of the Zeldovich and
Fisher equations with initial conditions in 
$[0,1]^N$ stay in $[0,1]^N$, 
provided the time-step is small enough.
 
The proof below shows that 
this step size does not depend on the initial conditions, 
it only depends on $\Delta$, i.e. the graph, and 
the functions 
\begin{equation}
\label{nonlinearities}
f_F(x) = (1-x)x, \quad f_Z(x) = (1-x)x(x - \alpha) 
\end{equation}
in the 
equations. The same comment applies to 
Lemma \ref{comparison-Zel-Fis-Euler}.

\begin{proof} 
Consider the first step of the iteration for the Zeldovich
equation, starting with two initial conditions 
$ x^M(0) < y^M(0) $ in $[0,1]^N$.

We have
\begin{eqnarray}
y^M(\Delta t) - x^M(\Delta t) & = & y^M(0) - x^M(0) + \notag \\
&   &  \Delta t [ \epsilon \Delta (y^M(0) - x^M(0)) + 
f_Z(y^M(0)) - f_Z(x^M(0))]. \notag
\end{eqnarray}
Examining the components of the 
$y^M(\Delta t)$, $x^M(\Delta t)$ we have that for every 
$k \in \{1, \ldots, N \}$,   
\begin{equation}
\label{Euler-comp-step-1}
y_k^M(\Delta t) - x_k^M(\Delta t) \geq 
[1 +  \Delta t (- \epsilon n_k + f_Z'({\tilde x}_k))]
(y_k^M(0) - x_k^M(0)),  
\end{equation}
where 
${\tilde x}_k \in [0,1]$, and  
$ - n_k = \Delta_{k,k}$. 

To maintain the 
$ y_k^M(\Delta t) - x_k^M(\Delta t)$ 
positive it  
suffices that 
\begin{equation}
1 +  \Delta t ( - \epsilon n_{max} + 
\min_{x \in [0,1]} f_Z'(x) > 0, 
\end{equation}  
with $n_{max} = \max_{k \in \{1, \ldots, N\}} n_k$.
This can be achieved for $M$ sufficiently large, 
and independent of $y^M(0)$, $x^M(0)$.

Applying this argument to the case where either  
$x^M(0) = {\bf 0}$, or $y^M(0) = {\bf 1}$, 
both static solutions of the Zeldovich equation, 
we then have 
$ {\bf 0} \leq x^M(\Delta t) < y^M(\Delta t) < {\bf 1}$, 
which also implies $ x^M(0)$, $y^M(0) \in [0,1]^N$.
We can iterate the argument for all remaining steps, 
with the same step size $T / M$.
The Fisher case is treated similarly. 
\end{proof}

Similarly we compare Euler approximates of
the Zeldovich and Fisher equations. 
We see that the Fisher approximations propagate
faster. The proof also shows that 
$(s(t),i(t))$ remains in $\tau^N$ for all times.

\begin{lemma}
\label{comparison-Zel-Fis-Euler2}
Let $T > 0$, and let $M > 1$ be an integer. 
Let $x_Z^M$, $x_F^M$ be the Euler approximations with
time-step ${\Delta t} = T /M $ over the interval $[0,T]$
of the Zeldovich (\ref{zeldovich})
and Fisher (\ref{fisher}) equations respectively.
Assume that $x_Z^M(0)  \leq x_F^M(0)$, with 
$x_Z^M(0)$, $x_F^M(0)$ in  $[0,1]^N$. Then 
$x_Z^M(m {\Delta t}) < x^F(m {\Delta t})$, 
$\forall m \in \{1, \ldots, M \}$, 
provided $M$ is sufficiently large.  
\end{lemma}

\begin{proof}
The argument is similar to the one used for
Lemma \ref{comparison-Zel-Euler} above and some details are omitted. 
Consider the first step of the iteration for the Zeldovich
and Fisher equations, starting with respective 
initial conditions 
$ x_Z^M(0) < x_F^M(0)$ in $[0,1]^N$, 
and let $f_Z$, $f_F$ denote the Zeldovich and 
Fisher nonlinearities respectively. 
We have 
\begin{eqnarray}
x_F^M(\Delta t) - x_Z^M(\Delta t) & = & x_F^M(0) - x_Z^M(0)  + \notag \\
& & \Delta t [ \epsilon \Delta (x_F^M(0) - x_Z^M(0)) + 
f_F(x_F^M(0)) - f_Z(x_Z^M(0))]. \notag
\end{eqnarray}
From  
\begin{eqnarray}
\label{Fis-Zel-difference}
f_F(x_F^M(0)) - f_Z(x_Z^M(0)) & = & 
[f_F(x_F^M(0)) - f_F(x_Z^M(0)] + \notag \\
&   & [f_F(x_Z^M(0)) - f_Z(x_Z^M(0))],  
\end{eqnarray}
and 
$$ (1-x)x \geq (1-x)x(x - \alpha), \quad \alpha \in (0,1),  
\forall x \in [0,1],$$ 
the second expression in (\ref{Fis-Zel-difference}) is a positive vector. 
Collecting the analogues of the
(\ref{Euler-comp-step-1}) for the components
of $x_F^M(\Delta t) - x_Z^M(\Delta t )$ we 
then have
\begin{equation}
x_F^M(\Delta t) - x_Z^M(\Delta t) \geq 
[1 +  \Delta t (- \epsilon  n_{max} + 
\min_{{x}\in [0,1]} f_F'(x))]
(x_F^M(0) - x_Z^M(0)). \notag  
\end{equation}
We can then take $M$ sufficiently large and independent
of the inital conditions so that
$ {\bf 0} \leq x_F^M(\Delta t) <  x_Z^M(\Delta t) 
\leq {\bf 1}$, and repeat the argument for all  
steps. 
\end{proof}

\begin{lemma}
\label{comparison-KmK-Fis-Euler}
Let $T > 0$, and let $M > 1$ be an integer.
Let $x^M$, and $(s^M,i^M)$ be the Euler approximations with
time-step ${\Delta t} = T /M $ over the interval $[0,T]$ of the 
Fisher (\ref{fisher}) and Kendrick-McKormack (\ref{km3}) 
equations respectively.
Assume that $i^M(0)  \leq x^M(0)$, with 
$x^M(0)$ in $[0,1]^N$, $(s^M(0),i^M(0)) \in \tau^N$.
Then 
$i^M(m {\Delta t}) < x^M(m {\Delta t})$, 
$\forall m \in \{1, \ldots, M \}$, 
provided $M$ is sufficiently large.  
For such $M$ we also have 
$(s^M(m {\Delta t}),i^M(m {\Delta t})) \in \tau^N$, 
$\forall m \in \{1, \ldots, M \}$.
\end{lemma}

\begin{proof}
The argument is similar to the one used for
Lemmas \ref{comparison-Zel-Euler},  
\label{comparison-Zel-Fis-Euler}
and some details are omitted. 
Consider the first step of the iteration for the 
Kendrick-McKormack
and Fisher equations, starting with respective 
initial conditions 
$x_F^M(0)$ in $[0,1]^N$, $(s^M(0),i^M(0)) \in \tau^N$.
We have at each site $k$
\begin{eqnarray}
x_k^M(\Delta t) - i_k^M(\Delta t) & = & x_k^M(0) - i_k^M(0) + \notag \\
& &  \Delta t [ \epsilon (\Delta (x^M(0) - i^M(0)))_k + \notag \\
& &  x_k^M(0)(1-x_k^M(0)) - i_k(0) s_k(0) + \gamma i_k^M(0)]. \
\end{eqnarray}
By $s^M_n(0) + i^M_n(0) \leq 1$, we therefore have 
\begin{eqnarray}
x_k^M(0)(1-x_k^M(0)) - i_k(0) s_k(0)  & \geq &  
x_k^M(0)(1-x_k^M(0)) - i_k(0) (1 - i_k^M(0))  + \notag \\
& &  f_F'({\tilde x}_k)(x_k^M(0) -  i_k^M(0) ) 
\end{eqnarray}
for $ {\tilde x}_k $ in $[i_k^M(0), x_k^M(0)] \subset [0,1]$. 
Then we have  
\begin{equation}
\label{fis-KmK-difference-1}
i^M(\Delta t) - x^M(\Delta t) \geq 
[1 + \Delta t (- \epsilon  n_{max} + 
\min_{{x}\in [0,1]} f_F'(x))]
(x^M(0) - i^M(0)),
\end{equation}
and therefore 
$i^M( {\Delta t}) < x^M({\Delta t})$
for $M $ sufficiently large and independent 
of the initial conditions.
 
A similar argument is used to show that $s^M(\Delta t)$, 
$i^M(\Delta t) > 0$ for $M$ sufficiently large and independent 
of the initial conditions, and we omit the details. 

By Lemma \ref{comparison-Zel-Euler}
similarly have $ x^M (\Delta t) \in [0,1]^N $ 
for $M$ sufficiently large
and independent of the initial conditions.
Finally adding the Euler formulas for $s^M({\Delta t})$, 
$i^M({\Delta t})$, and using $\Delta {\bf  1} = 0$, we have
\begin{eqnarray}
\label{fis-KmK-difference-2}
{\bf 1 } - (s^M({\Delta t}) + i^M({\Delta t})) & =  & 
{\bf 1 } - (s^M(0) + i^M(0) + \notag \\
& & {\Delta t}[\epsilon ({\bf 1 } - (s^M(0) + i^M(0)) 
+ \gamma   i^M(0)  ] \notag \\
& \geq & (1 - {\Delta t} \epsilon d_{max})[{\bf 1 } - (s^M(0) + i^M(0)],  
\end{eqnarray}
which is positive for $M$ sufficiently large
and independent of the initial conditions. 
We can then iterate the argument for the remaining steps. 
\end{proof}

The fact that the convergence 
of the approximate solutions 
$x^M$ of the Euler method to the trajectory $x$ in 
Lemma \ref{Euler-conv} preserves the 
partial order follows from the following.

Let  
$\{g^M\} = \{ g_M \}^{\infty}_{M = 1}$, 
denote a sequence 
of arrays 
(of increasing size $M +1$) 
$g^M = (g^M_0, \ldots, g^M_M)$ of vectors 
$g_m^M \in {\bf R}^K$, $m = 1, \ldots, M$.
Let  $ \{g^M\} \rightarrow g $ denote that  
\begin{equation}
\label{unif-conv}
\lim_{M \rightarrow \infty} 
\max_{m \in \{0, \ldots, M\} } 
||g^M_m  - g( m ({T / M}) )|| = 0,
\end{equation}
where  
$g:[0,T] \rightarrow {\bf R}^K$, 
and $|| \cdot ||$ is the norm in 
${\bf R}^K$.

\begin{lemma}
\label{monotone-limits}
Consider sequences
$\{ g^M \}$, $\{h^M \}$ as above satisfying that
for all $M $ sufficiently large
we have that $g^M_m < h^M_m$, 
$\forall m \in \{0, \ldots, M\}$. 
Suppose also that there exist 
continuous functions $g$, $h:[0,T]\rightarrow {\bf R}^K$ 
for which 
$\{ g^M \} \rightarrow g$, and $\{h^M \} \rightarrow h$
respectively. Then $g(t) \leq h(t)$, 
$\forall t \in [0,T]$.
\end{lemma}

\begin{proof}
The statement follows from continuity of $h - g$ 
in $[0,T]$, since it is easy to see that $h(t_0) - g(t_0) > 0$ for some 
$t_0  \in (0,T]$ and the convergence leads to a contradiction. 
\end{proof}

We now prove Proposition \ref{compare-Zel-Zel}. 

\begin{proof}
Combining 
the comparison of approximate solutions  
produced by the Euler method 
Lemma \ref{comparison-Zel-Euler},
with the approximation 
Lemmas \ref{Euler-meth-lem}, \ref{monotone-limits} 
we have that $x(0) < y(0)$ implies 
$x(t) \leq y(t)$, $\forall t \in [0,T]$. 
To show the strict inequality we use the fact that the 
evolution can be also defined uniquely also backwards in time. 
Thus $x(t) \leq y(t)$, for some $t \in (0,T]$ leads
to a contradiction.  
\end{proof}

Propositions \ref{compare-Zel-Fis}, 
\ref{compare-Fis-KmK} follow in the same way, 
but we can not apply the backwards evolution argument,
and do not have strict inequality. 

\section{Numerical results on front propagation}

In the first part of this section we solve numerically 
the Zeldovich system (\ref{zeldovich}) for initial conditions 
that are near the computed static solutions. We confirm
the results of Section 3 on thresholds, and examine the evolution 
of the front-like initial conditions 
for couplings that are above the threshold.
In the second part we verify some of the 
predictions of the comparison results of Section 4. One main 
observation is that the propagation 
of the Fisher, and Kermack-McKendrick models is much faster
than the ones seen in the Zeldovich models. 
A third part examines the propagation of Zeldovich fronts for larger
values of the local excitation threshold $ a $. We 
can then observe very rapid front propagation. 
In all simulations below  
we use the variable step 5-6 dopri5 solver of Hairer and 
Norsett \cite{Hairer} in double precision with a relative tolerance 
of $10^{-10}$.

\subsection{Time evolution of the Zeldovich fronts}

Branches of static solutions $x(\epsilon)$ are labeled by 
the corresponding $\epsilon \rightarrow 0$ limit $x(0)$, obtained 
by decreasing $\epsilon$. The value of $\epsilon$ at the fold 
is $\epsilon_0(x(0))$ (or simply $\epsilon_0$ when the branch in question 
is clear).  
In addition to the numerical and theoretical $\epsilon_0$ values 
from Section 3, we here obtain a third estimate of $\epsilon_0$ by 
integrating (\ref{zeldovich}) starting with $\epsilon < \epsilon_0$ 
and increasing $\epsilon$ slowly on each run. 
The typical behavior is the following.  For small $\epsilon$ we 
always find a static solution. As $\epsilon$ is increased past 
a threshold $\epsilon_0$, the solution destabilizes and gives rise to the 
homogeneous flat state $[1, \dots, 1]$. This estimated 
$\epsilon_0$ is between the 
largest $\epsilon$ leading to convergence to a similar static front, 
and the smallest $\epsilon_0$ leading to a trajectory that diverges 
from the front. The three estimates of $\epsilon_0$
are given in Table \ref{tab1} for some examples, and confirm the results
expected from Section 3.

We now present some examples of the evolution slightly above 
the $\epsilon_0$ for different configurations.  We examine
how the connectivity of a node influences the destabilization of a
front centered at that node.

% node6  

In the first example 
we consider the evolution from an initial condition near 
a static front localized at node 6 which (with connectivity 1).
The front belongs to the branch of $[0,0,0,0,0,1]^T$, which is expected 
to be last to be destabilized.
We use $ \epsilon=3.0 ~ 10^{-3}$, 
slightly above the computed threshold $\epsilon_0 = 2.998 ~ 10^{-3}$,  
and the initial condition 
\begin{equation}
\label{node6-initcond}
[1.536 ~ 10^{-3},  8.680 ~ 10^{-5} , 1.537 ~ 10^{-3},  1.578 ~ 10^{-3},
5.350 ~ 10^{-2}, 0.997]^T.
\end{equation}
Notice how it decays very rapidly from the node 5 to the nodes 
1 and 4 then node 3. 
The evolution is shown in Fig. \ref{znode6}, we see that the wave goes 
successively from 5 ,4 , 3 ,1 and 2.
\begin{figure} [H]
\centerline{
\epsfig{file=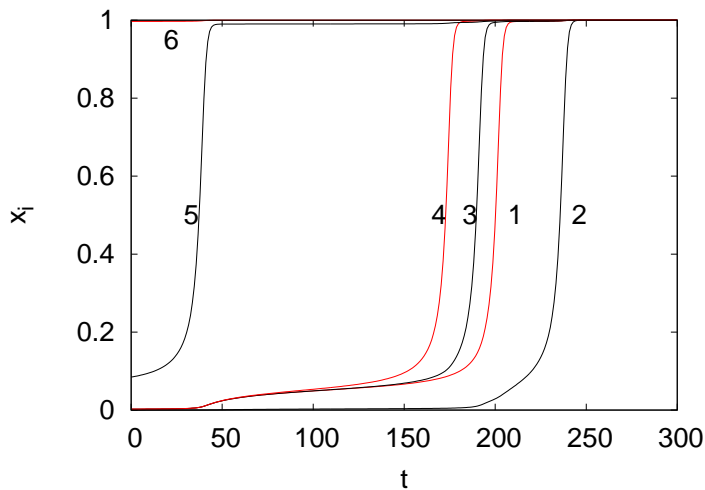,width=0.8\linewidth,angle=0}
}
\caption{Time evolution of the different nodes for an initial front
solution centered on node 6, of connectivity 1, with  
$\epsilon = 3.0 ~ 10^{-3}$ and the 
initial condition \ref{node6-initcond}. 
}
\label{znode6}
\end{figure}

% node 2

In the second example we consider an initial condition
near the static solutions of the branch $[0,1,0,0,0,0]^T$.
We use $ \epsilon= 2.6875 ~ 10^{-3}$, 
which is slightly above computed $ \epsilon_0 = 2.547 ~ 10^{-3}$ 
for the branch, 
and the initial condition 
\begin{equation}
\label{node2-initcond}
[5.088 ~ 10^{-2},  0.994, 4.298 ~ 10^{-2},   1.165 ~ 10^{-3},  
2.35 ~ 10^{-3}, 6.132 ~ 10^{-5}]. 
\end{equation}
The value $ x_1 = 5.088 ~ 10^{-2}$ is very close 
to $a/2$ which is the value observed by the
continuation method for $\epsilon=\epsilon_c$. 
The evolution is shown in Fig. \ref{znode2}.
The solution destabilizes following the fixed point so $x_1$ and $x_3$
remain close to $a/2$ for a long time before going to 1. 
We see that the wave follows the
connectivity as it propagates from node 1 (3) to node 3 (4). Then  node
5 (4) destabilizes 
and finally node 4. Node 6 is just destabilizing for $t=900$. There are then
different time scales in the dynamics depending on the connectivity.

We also note that $\epsilon $ is greater than the threshold 
$\epsilon_0 = 2.5 ~ 10^{-3}$ for the branch $[1,1,1,1,1,0]^T$ (see Section 3), 
this means that the front will not stop at node 5, it will also 
destabilize node 6. 
\begin{figure} [H]
\centerline{
\epsfig{file=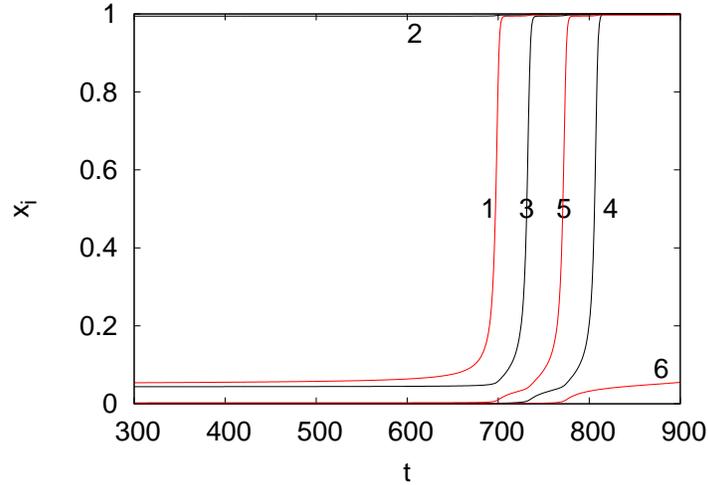,width=0.8\linewidth,angle=0}
}
\caption{Time evolution of the different nodes for an initial front
solution centered on node 2, of connectivity 2, 
$\epsilon = 2.6875 ~ 10^{-3}$, initial condition in \ref{node2-initcond}. 
 }
\label{znode2}
\end{figure}

% node 3

We now consider an initial condition centered on node 3,
near static solutions of the branch $[0,0,1,0,0,0]^T$.
We use $ \epsilon= 2.54 ~ 10^{-3} $, 
which is slightly above the computed threshold 
$ \epsilon_0 = 2.528 ~ 10^{-3} $ for the branch, 
and the initial condition 
\begin{equation}
\label{exp3-initcond}
[4.946~ 10^{-2},   5.404 ~ 10^{-2}, 0.989,       4.785~ 10^{-2},
4.206~ 10^{-2}, 1.05410^{-3} ].
\end{equation}
The evolution given in Fig. \ref{znode3} shows that
the front centered on node 3 of connectivity 4 destabilizes in the
same way as the one centered on node 2
except that now nodes 2,1,5 and 4 have values around $a/2$ for a long time. 
Node 6 will destabilize after a long time. As in the previous 
example $\epsilon $ is greater than the threshold $\epsilon_0 = 2.5 ~10^{-3}$
for the branch $[1,1,1,1,1,0]^T$. 
\begin{figure} [H]
\centerline{
\epsfig{file=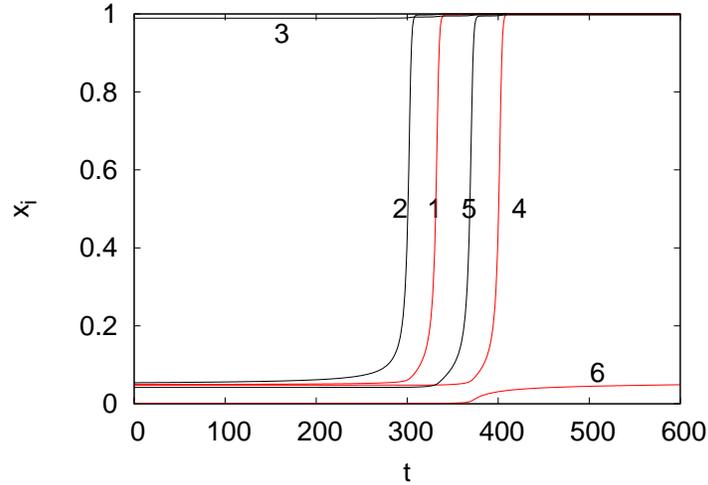,width=0.8\linewidth,angle=0}
}
\caption{Time evolution of the different nodes for an initial front
solution centered on node 3, of connectivity 4,
$ \epsilon= 2.54 ~ 10^{-3}$, initial condition in \ref{exp3-initcond}.}
\label{znode3}
\end{figure}

% node 123

We now consider initial conditions near 
static solutions of the branch $[1,1,1,0,0,0]^T$, 
see Fig. \ref{z123}. 
We use $\epsilon = 1.32~ 10^{-3}$, 
slighly above the critical $\epsilon_0 = 1.3103 ~10^{-3}$ 
(see Section 4), and the intial condition
\begin{equation}
\label{node123-initcond} 
[0.999  , 0.99999, 0.997 , 1.610 ~ 10^{-2} ,  4.980 ~ 10^{-2} ,  
6.485 ~ 10^{-4}]^T 
\end{equation}
The evolution is shown in Fig. \ref{z123}.
Node 5 is the first to destabilize, followed by node 4.
We also see that node 6 remains at 
its level because $\epsilon$ is smaller than the threshold $\epsilon_0$
for the static front of the type $[1 ,1 ,1,1 ,1 ,0]^T$.
\begin{figure} [H]
\centerline{
\epsfig{file=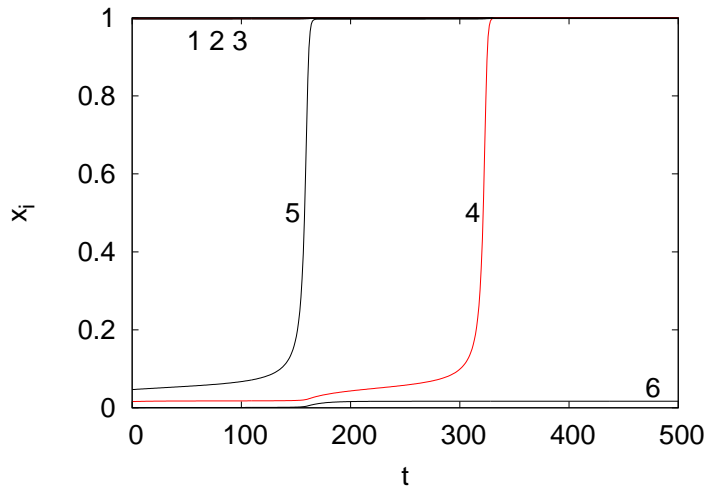,width=0.8\linewidth,angle=0}
}
\caption{Time evolution of the different nodes for an 
initial front solution of the type $[1,1,1,0,0,0]^T$, 
with $\epsilon = 1.32 ~ 10^{-3}$ 
and the initial condition \ref{node123-initcond}.
}
\label{z123}
\end{figure}

One can estimate the time for $x_5$ to grow, using the normal form displayed in
Fig. \ref{conti}  as a function of $\delta=\epsilon-\epsilon_c$. It gives
$${\dot x_5} = x_5^2 + \delta ,$$
so that
\begin{equation}
\label{center-eq}
x_5(t) = \sqrt{\delta}\tan( \sqrt{\delta}  t) ,
\end{equation}
which grows as $1/ \sqrt{\delta}$
We have 
$$\epsilon_c=1.3 10^{-3}, ~ \epsilon=1.4 10^{-3}, \delta = 10^{-4},~
1/\sqrt{\delta} = 100 .$$
From Fig. \ref{z123} one sees that the typical time of destabilization 
of $x_5$ is about 100 so the estimate is correct.

These results confirm that generalized static fronts exist for small 
$\epsilon$ and disappear for $\epsilon>\epsilon_c$; they 
are summarized in Table \ref{tab1}. 
\begin{table} \label{tab1}
\begin{tabular}
{|l | c | c | c |c | r|}
  \hline
connectivity &  node  & branch           &  $\epsilon_0$  &  $\epsilon_0$  & expression (\ref{epsofa}) \\ 
             &        &                  &  from time evolution&  continuation  &     \\ \hline
1            &    6   &  $(0 0 0 0 0 1)$ & $3.0~ 10^{-3}$  &   $2.998~ 10^{-3}$ & $2.97 ~10^{-3}$   \\ \hline
2            &    2   &  $(0 1 0 0 0 0)$ & $ 2.7~ 10^{-3}$ &   X  & $2.79 10^{-3}$ \\ \hline
4            &    3   &  $(0 0 1 0 0 0)$ & $2.54~ 10^{-3}$  &  X   &   $2.63 10^{-3}$ \\ \hline
2            &  1 2 3 &  $(1 1 1 0 0 0)$ & $1.31~ 10^{-3}$ & $1.310~ 10^{-3}$ &     $1.32 ~10^{-3}$  \\     \hline
 \end{tabular}
\caption{Critical $\epsilon$ for the "generalized front" to destabilize for 
different initial conditions.} 
\end{table}
The above examples also suggest a qualitative picture of the 
propagation of fronts, where one can use the 
analytical expresion (\ref{epsofa}) for $\epsilon_0$
to guess the order in which the differerent nodes 
are excited. It appears that given a
configuration of excited sites, the next site is the one in the 
neighborhood of the configuration that has the 
largest number of connections with the configuration connections. 
In the case where we have more than one such sites, 
the one that has the fewest connections, see e.g. the example  
of Fig. \ref{znode6}. 
This rule is consistent with the calculation of the 
smallest $\epsilon_0$ values from (\ref{epsofa}) among the 
possible $n_c$ in the vicinity of a configuration.
This rule does not include all posibilities, but it points to 
a possible connection between the $\epsilon_0$ for the various 
branches, and the propagation of the front.  
An estimation of $ \delta = \epsilon -  \epsilon_0$ leads to an approximate 
time for the site $n_c$ to be excited, using (\ref{center-eq}).

\subsection{Comparison between different solutions and front propagation 
models}

% Zeldovich /Zeldovich  \\

To illustrate the comparison of two initial conditions under 
the Zeldovich evolutions, we  
show solutions from initial conditions $[1,1,1,0,0,0]^T$ and
$[1,1,0,0,0,0]^T$ respectively. We use $\epsilon =1.4 ~ 10^{-3}$. 
The time evolution is indicated 
Fig. \ref{cz} where the nodes 4 and 5 are shown. 
The trajectories increase faster for the first initial 
condition than for the second.
\begin{figure} [H]
\centerline{
\epsfig{file=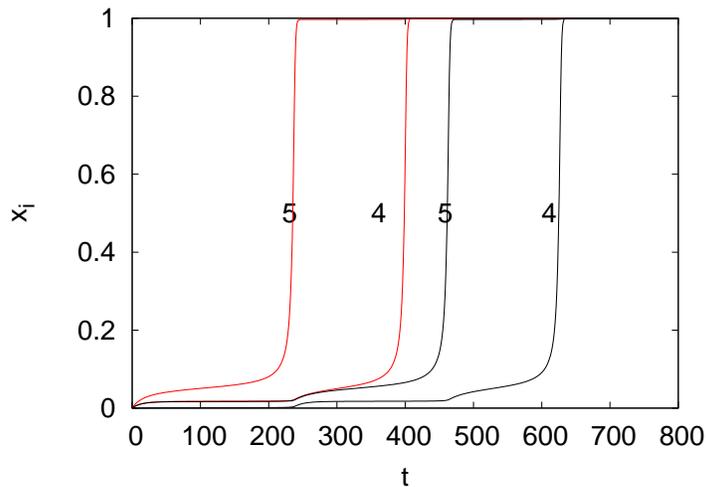,width=0.8\linewidth,angle=0}
}
\caption{Time evolution of nodes 4 and 5 for the initial fronts
$[1,1,1,0,0,0]^T$ in continuous line (red online) and 
$[1,1,0,0,0,0]^T$ in dashed line.}
\label{cz}
\end{figure}

% Fisher /Zeldovich  

To illustrate the comparison between trajectories of 
the Fisher and Zeldovich equations 
we use the initial condition $[1,1,0,0,0,0]^T$, 
with $\epsilon=1.4 ~10^{-3}$. It is presented 
in Fig. \ref{cfz}.
Note that the scale in time is much shorter than in Fig. \ref{cz} , here
for $t=20$
the front has invaded the graph. Therefore the Fisher solution
will always be larger than the Zeldovich one. Also the profile is different
since there are no fixed points other than the flat 1 homogeneous
state.
\begin{figure} [H]
\centerline{
\epsfig{file=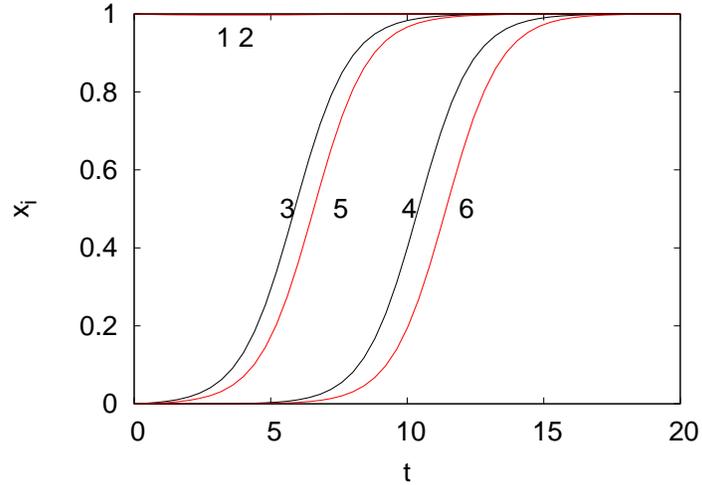,width=0.8\linewidth,angle=0}
}
\caption{Time evolution of the initial front
$[1,1,0,0,0,0]^T$.} 
\label{cfz}
\end{figure}

% Kermack-McKendrick /Zeldovich  \\

We also consider the evolution of the Kermack-McKendrick
model (\ref{km3}).
When the  decay term $\beta$ for the infected component $i$ is
zero, the evolution of $i$ is identical to the one of the
Fisher model (\ref{fisher}). This is because (\ref{km3}) conserves
$s+i$. For example taking as initial condition 
$$ 
s = [0, 0, 1, 1,  1,   1]^T, \quad i=[1, 1, 0, 0,  0,   0]^T $$ 
yields exactly the same dynamics for $i$ as the one of Fig. \ref{cfz}.
On the other hand, if we choose $s+i <1$ and still the same 
initial $i$, then the trajectories of (\ref{km3}) are below the
ones of (\ref{fisher}). Nevertheless the characteristic time for
the orbits of (\ref{km3}) to reach saturation is the same as for (\ref{fisher}).
When $\beta>0$ is small, the infected component reaches a maximum in  this
characteristic time and then decays over a time scale $1/\beta$
Fig.  \ref{kmk001}
shows the evolution of the infected component for $\beta=0.01$
and $\epsilon = 1.35 ~10^{-3}$
To see propagation on the network, $\beta$ should be smaller than 
the diffusion time $1/\epsilon$.
\begin{figure} [H]
\centerline{
\epsfig{file=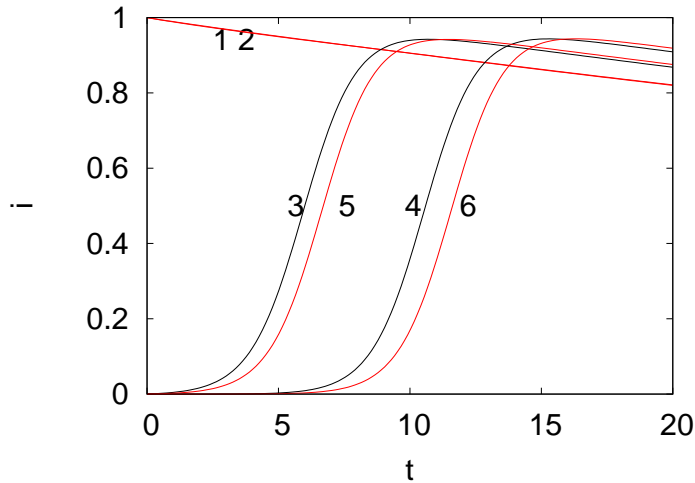,width=0.8\linewidth,angle=0}
}
\caption{Time evolution of the infected component $i$ for the initial front
$s = [0, 0, 1, 1,  1,   1]^T,~i=[1, 1, 0, 0,  0,   0]^T$ for the
Kermack-Mackendrick model (\ref{km3}) with $\beta=0.01$.
}
\label{kmk001}
\end{figure}
The comparisons between the Zeldovich models on the one hand, and 
the Fisher, and the Kermack-McKendrick
models show that the later two lead to a much faster propagation. 
This makes the comparison between the Fisher, and Kermack-McKendrick 
models a more interesting result.

The examples above suggest also that   
the order in which the different nodes become excited
in the three models is the same. 
This order seems to depend only on the geometry  
of the graph.
It may be possible to use different (possibly 
branch or site dependent) parameters 
$\epsilon$, $\gamma$, and $a$ 
for the Zeldovich and Fisher systems to make the 
propagation speeds comparable.

 \subsection{Influence of the parameter $a$}

To conclude this numerical section, we consider how the fixed points of
the Zeldovich equation and its dynamical solutions depend on the parameter
$a$. To illustrate how $a$ changes the fixed point and it's subsequent 
destabilization, we consider the front centered on node 6 of the
type $[0 0 0 0 0 1]^T$. For $a=0.3$ and $\epsilon = 1.25 ~ 10^{-3}$
we obtain the static front 
\be\label{a03_node6_stat}
[0.104,  5.477~10^{-2} , 0.107,  0.124, 0.295,  0.852]^T.
\ee
Compared to the one for $a =0.1$ (\ref{node6-initcond}), this
front is much broader. Here we see that $x_5 \approx aĥ$ and
$x_3,x_4$ and $x_1$ are close to $a/2$. 
\begin{figure} [H]
\centerline{
\epsfig{file=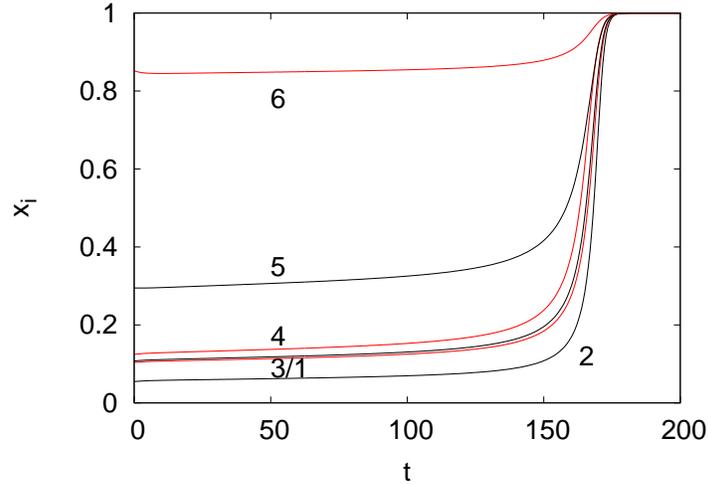,width=0.8\linewidth,angle=0}
}
\caption{Time evolution for the initial condition 
$(0.1036,  5.477~ 10^{-2},  0.107 , 0.124, 0.294, 0.851 )$ 
for $\epsilon = 1.3 ~ 10^{-1}$ close to the critical value
$\epsilon_0 = 1.25 ~ 10^{-1}$ for the Zeldovich equation.
The parameter $a=0.3$. 
}
\label{a03_6}
\end{figure}
The time evolution of the initial condition (\ref{a03_node6_stat})
is presented in Fig. \ref{a03_6}. Note the large velocity with which
the front "invades" the network. For $a=0.1$, in Fig. \ref{znode6}
we had a well separated dynamics of node 5 which destabilized first.
Here we cannot distinguish the evolution of node 5 from the one of
the other nodes. Since the front is much wider, it averages out the
network and propagates much faster.

Because the front becomes very wide, the 
formula (\ref{epsofa}) will underestimate the
critical $\epsilon$. 
Table \ref{tab2} shows $\epsilon_0$ for $a=0.1,~0.2 $ and $0.3$
obtained for the static solution centered on node 6.
\begin{table} \label{tab2}
\begin{tabular}
{|l | c | r|}
  \hline
$a$   &     $\epsilon_0$  &      expression (\ref{epsofa}) \\ \hline
0.1   & $3.~ 10^{-3}$     &        $2.8 ~ 10^{-3}$  \\ \hline
0.2   & $1.6 ~ 10^{-2} $  & $1.28 ~ 10^{-2} $    \\ \hline
0.3   & $ 1.25~ 10^{-1} $ & $3.48 ~ 10^{-2} $\\ \hline
 \end{tabular}
\caption{Critical $\epsilon$ for the "generalized front" centered on node 6 
to destabilize for different values of $a$. }
\end{table}
As expected (\ref{epsofa}) underestimates $\epsilon_0$ as $a$ increases. It 
gives the right order of magnitude for $a=0.2$ but is clearly wrong for
$a=0.3$.

\section{Conclusion}

We studied analytically and numerically 
a bistable reaction diffusion on an arbitrary finite network. 
We show that stable static fronts 
exist everywhere on the network for small diffusivity. 
We give the asymptotics of these fixed points and derive from them
a simple depinning criterion which is validated both by continuation
techniques and by solving the time dependent problem.
The justification of the depinning criterion is an open problem, 
and may be related to the small value of the local excitation parameter 
$a$.
The numerical simulations suggest that the moving front "feels"
the different static configurations, as it travels accross the network.

We also compare different solutions of the Zeldovich model and show 
how "large" fronts dominate "small" fronts in the 
dynamics. The time dependent solutions of the Fisher and Kermack-Mckendrick
original models are compared to the ones of the Zeldovich; they have
a much shorter time scale and no treshold. This effect might
be expected from the instability of the origin in 
the Fisher and Kermack-Mckendrick models.
This seems to reduce their interest as opposed to the Zeldovich model.
On the other hand all three models describe qualitatively
similar front expansion scenarios above the Zeldovich threshold.
Another posibility is that 
The behavior of the Zeldovich model
below the highest branch threshold may reflect some pinning phenomena 
related to epidemics.  

Finally we investigate numerically 
larger local excitation thresholds and show that fronts 
become wider and travel much faster across the network.

{\bf Acknowledgements}  \\
J.G. C. thanks the Universidad Nacional Aut\'onoma de M\'exico
for its hospitality during two visits.
The work of J.G. C. is supported partially by a grant from
the Grand Reseau de Recherche, Transport Logistique et Information
of the Haute-Normandie region. 
The authors acknowledge the Centre de Ressources Informatiques de Haute
Normandie for the computations together with Ana Perez and Ramiro Chavez 
from IIMAS UNAM for technical support.


\begin{thebibliography}{99}

\bibitem{scott}
A. C. Scott
``Nonlinear science, emergence and dynamics of coherent structures'',
 Oxford University Press (2003).

\bibitem{nabarro} F. R. N. Nabarro, "Dislocations in a simple cubic lattice",
Proc. Roy. Soc. London, {\bf 59}, 256-272, (1947).


\bibitem{cks11} 
J.-G. Caputo, A. Knippel and E. Simo, "Oscillations of simple networks", 
J. Phys. A: Math. Theor. 46, 035100 (2013) \\
http://arxiv.org/abs/1109.3071

\bibitem{CB03} A. Carpio and L. Bonilla , "Depinning transitions
in discrete reaction-diffusion equations", SIAM J. Appl. Math.
63, 1056-1082, (2003).

\bibitem{CDEMLPAV09}
G. Cruz-Pacheco, L. Duran, L. Esteva, A.A. Minzoni, M.Lopez-Cervantes, 
P. Panayotaros, A. Ahued-Ortega, I. Villase\~nor Ruiz, 
Modelling of the influenza A(H1N1) outbreak in Mexico City, 
April-May 2009, with control measures, Eurosurveillance 14, 26 (2009)


\bibitem{minoux} M. Gondran and M. Minoux, "Graphs and Algorithms", 
John Wiley and Sons, (1984).



\bibitem{EN93} T. Erneux, G. Nicolis, 
Propagating fronts in discrete bistable reaction-diffusion
system, Physica D 67, 237-244 (1993)

\bibitem{HM10} A. Hoffman and J. Mallet-Paret, "Universality of 
Crystallographic Pinning", J. Dyn. Diff. Equat.  22, 79-119, (2010).


\bibitem{crs01} D. Cvetkovic, P. Rowlinson and S. Simic,
 "An Introduction to the Theory of
Graph Spectra",  London Mathematical Society Student Texts (No. 75), (2001).

\bibitem{km27}
W. O. Kermack and A. G. McKendrick,
``A Contribution to the Mathematical Theory of Epidemics'',
Proc. Roy. Soc. Lond. A 115, 700-721 (1927).


\bibitem{Hairer}  E. Hairer, S. P. Norsett and G. Wanner. {\it Solving
ordinary differential equations I}, Springer-Verlag, (1987).

\bibitem{I08} A. Iserles, 
{\it First Course in the Numerical Analysis of Differential Equations, 
2nd Edition},
Cambridge University Press
Cambridge (2008)

\bibitem{keener87}
J.P. Keener, Propagation and its failure 
in coupled systems of discrete excitable cells, 
SIAM J. of Appl. Math. 47, 556-572 (1987)


\bibitem{K77}
H.B. Keller, Numerical solution of bifurcation and nonlinear 
eigenvalue problems, 
in P. H. Rabinowitz, editor, Applications of Bifurcation Theory, 
Academic Press (1977)


\bibitem{MA94} R.S. MacKay, S. Aubry: Proof of existence of breathers for
time-reversible or Hamiltonian networks of weakly coupled
oscillators, Nonlinearity  7, 1623-1643 (1994)

\bibitem{P70} 
M.J.D. Powell, 
A hybrid method for nonlinear equations, in 
Numerical methods for nonlinear algebraic equations, 
P. Rabinowitz, ed., Gordon and Breach, New York (1970) 



\bibitem{Z86} E. Zeidler, Nonlinear Funcional Analysis and
its Applications I, Springer, New York (1986)

\bibitem{matlab} The Mathworks \\
http://www.mathworks.com

\end{thebibliography}
\end{document}